\documentclass[reprint,amsfonts, amssymb, amsmath,  showkeys, nofootinbib, prx, superscriptaddress, twocolumn,longbibliography]{revtex4-2}
\usepackage{float}
\makeatletter
\let\newfloat\newfloat@ltx
\makeatother
\usepackage[english]{babel}

\usepackage[utf8]{inputenc}
\usepackage{graphics}
\usepackage{selinput}
\usepackage[normalem]{ulem}
\usepackage[shortlabels]{enumitem}

\usepackage{braket}
\usepackage{amsthm}
\usepackage{mathtools}
\usepackage{physics}
\usepackage{xcolor}
\usepackage{graphicx}

\usepackage{placeins}
\usepackage[T1]{fontenc}
\usepackage{lipsum}
\usepackage{bm}

\usepackage[linesnumbered,ruled,vlined]{algorithm2e}
\SetKwInput{kwInit}{Init}

\def\HC{\mathcal{H}}

\def\LC{\mathcal{L}}

\usepackage{mathrsfs} 

\def\ad{^{\dagger}}

\def\a{\alpha}

\newcommand{\fsnull}[1]{}
\newcommand{\old}[1]{}

\usepackage[toc,page]{appendix}
\usepackage[colorlinks=true,citecolor=blue,linkcolor=magenta]{hyperref}

\newcommand{\dya}[1]{\ket{#1}\!\bra{#1}}

\newcommand{\NC}{\mathcal{N}}
\newcommand{\OC}{\mathcal{O}}
\newcommand{\PC}{\mathcal{P}}

\newcommand{\SC}{\mathcal{S}}

\renewcommand{\geq}{\geqslant}
\renewcommand{\leq}{\leqslant}

\DeclareMathOperator*{\argmin}{arg\,min}
\renewcommand{\vec}[1]{\boldsymbol{#1}}  

\newcommand*{\id}{\openone}

\newcommand{\bs}{\textsf{BS}}

\renewcommand{\th}{\theta } 

\newcommand{\SWAP}{{\rm SWAP}}

\def\be{\begin{equation}}
\def\ee{\end{equation}}
\def\bs{\begin{split}}
\def\e{\end{split}}
\def\ba{\begin{eqnarray}}
\def\bea{\begin{eqnarray}}

\def\tea{\end{eqnarray}}
\def\ea{\end{eqnarray}}
\def\eea{\end{eqnarray}}

\def\SU{\text{SU}}

\def\a{\alpha}
\def\b{\beta}

\def\a{\alpha}

\def\gam{\gamma}

\def\b{\beta}

\def\SU{\text{SU}}

\def\a{\alpha}
\def\b{\beta}

\newtheorem{theorem}{Theorem}
\newtheorem{lemma}{Lemma}
\newtheorem{corollary}{Corollary}

\newtheorem{proposition}{Proposition}

\usepackage{amssymb}
\usepackage{dsfont}

\def\be{\begin{equation}}
\def\te{\end{equation}}
\def\ee{\end{equation}}
\def\ba{\begin{eqnarray}}
\def\bea{\begin{eqnarray}}

\def\tea{\end{eqnarray}}
\def\ea{\end{eqnarray}}
\def\eea{\end{eqnarray}}

\newcommand{\beq}{\begin{equation}}
\newcommand{\eeq}{\end{equation}}

\begin{document}

\title{Analyzing the free states of one quantum resource theory as resource states of another}

\author{Andrew E. Deneris}
\affiliation{Information Sciences, Los Alamos National Laboratory, Los Alamos, NM 87545, USA}
\affiliation{Quantum Science Center, Oak Ridge, TN 37931, USA}

\author{Paolo Braccia}
\affiliation{Theoretical Division, Los Alamos National Laboratory, Los Alamos, NM 87545, USA}

\author{Pablo Bermejo}
\affiliation{Information Sciences, Los Alamos National Laboratory, Los Alamos, NM 87545, USA}
\affiliation{Donostia International Physics Center, Paseo Manuel de Lardizabal 4, E-20018 San Sebasti\'an, Spain}
\affiliation{Department of Applied Physics, University of the Basque
Country (UPV/EHU), 20018 San Sebastián, Spain}

\author{N. L. Diaz}
\affiliation{Information Sciences, Los Alamos National Laboratory, Los Alamos, NM 87545, USA}
\affiliation{Center for Non-Linear Studies, Los Alamos National Laboratory, 87545 NM, USA}

\author{Antonio A. Mele}
\affiliation{Dahlem Center for Complex Quantum Systems, Freie Universität Berlin, 14195 Berlin, Germany}
\affiliation{Theoretical Division, Los Alamos National Laboratory, Los Alamos, NM 87545, USA}

\author{M. Cerezo}
\thanks{cerezo@lanl.gov}
\affiliation{Information Sciences, Los Alamos National Laboratory, Los Alamos, NM 87545, USA}
\affiliation{Quantum Science Center, Oak Ridge, TN 37931, USA}

\begin{abstract}
In the context of quantum resource theories (QRTs), free states are defined as those which can be obtained at no cost under a certain restricted set of conditions. However, when taking a free state from one QRT and evaluating it through the optics of another QRT, it might well turn out that the state is now extremely resourceful. Such realization has recently prompted numerous works characterizing states across several QRTs. In this work we contribute to this body of knowledge by analyzing the resourcefulness in free states for--and across witnesses of--the QRTs of multipartite entanglement, fermionic non-Gaussianity, imaginarity, realness, spin coherence, Clifford non-stabilizerness, $S_n$-equivariance and non-uniform entanglement. We provide rigorous theoretical results as well as present numerical studies that showcase the rich and complex behavior that arises in this type of cross-examination. 
\end{abstract}

\maketitle
\section{Introduction}

Quantum resource theories (QRTs)~\cite{chitambar2019quantum,brandao2015reversible} provide a theoretical framework to study setups where only a subset of quantum evolutions are allowed (or ``free''), and only a subset of states can be prepared (i.e., are considered to also be ``free''). Crucially, QRTs can be intimately tied to the task of simulating certain quantum systems under restricted types of operations, providing operational meaning to this mathematical paradigm. For instance, one can see the connection between QRTs and classical simulation methods through the prototypical examples of multipartite  entanglement~\cite{horodecki2013quantumness,chitambar2019quantum,bennett1996concentrating,vedral1997quantifying,bennett1999quantum}, fermionic non-Gaussianity~\cite{lami2018guassian, takagi2018convex, zhuang2018resource, chitambar2019quantum,denzler2024learning,bittel2025pac,wan2022matchgate,mele2024efficient,oszmaniec2022fermion}, and  Clifford non-stabilizerness~\cite{veitch2014resource,howard2017application,chitambar2019quantum, leone2022stabilizer}, as their associated QRTs lead to irriguous grounding to the problem  of simulating low-entanglement states, near-Gaussian states, or low-magic states via tensor networks~\cite{orus2014practical}, Wick's theorem~\cite{valiant2001quantum, knill2001fermionic,terhal2002classical,bravyi2004lagrangian,divincenzo2005fermionic,somma2006efficient,jozsa2008matchgates}, or the Gottesman-Knill theorem~\cite{gottesman1998heisenbergrepresentation, aaronson2004improved, nest2008classical}; respectively. 

While the specific in's and out's of each QRT tend to be studied separately, there has been a tremendous interest in cross-examining the free operations and states of a given QRT through the optics of another. For instance,  such analysis could seek to identify and characterize states, or families thereof, which possess low amounts of resource in more than one QRT, and which could lead to hybrid simulation methods that capture evolutions beyond those simulable by techniques based on a single type of resource~\cite{lami2023nonstabilizerness,haug2023quantifying,masot2024stabilizer,mello2024hybrid,mello2025clifford,goh2023lie,gu2024doped,qian2024augmenting,nakhl2025stabilizer,qian2025clifford,dowling2024magic,paviglianiti2024estimating,andreadakis2025exact,qian2024augmenting,fux2024disentangling}. The previous has led to a veritable Cambrian explosion of mixing and matching different measures of resourcefulness in all kinds of physical systems, and sets of QRT's free states~\cite{tarabunga2024critical,sarkis2025molecules,deside2025detecting,deside2025detecting,moca2025non,tarabunga2025efficient,jasser2025stabilizer,cianciulli2024bipartite,gigena2021many,collura2024quantum,gu2025magic}. Here it is worth noting that up to this point most works tend to focus on two QRTs at a time, studying for instance the non-stabilizerness (or magic) in Gaussian states, or the fermionic Gaussianity in low-entangled ground states of spin-systems representable by matrix product state techniques.  These analyses, while extremely important, have the potential downside of leading to a patch-worked understanding of how different types of resources can coexist in the same state.

In this work, we seek to contribute to the body of knowledge of analyzing one QRT's free states as resource states of another (see Fig.~\ref{fig:schematic}) by simultaneously focusing on eight QRTs. Our goal it to provide a more holistic and comprehensive view of the characterization of a state's resourcefulness. We begin our work in Section~\ref{sec:framework} by first presenting a general framework for QRTs, as well as how to quantify their resourcefulness via group Fourier harmonic analysis-based purity-type witnesses (for additional details, we refer the reader to our companion manuscript~\cite{bermejo2025characterizing}). Next, in Section~\ref{sec:QRTs} we  give a comprehensive introduction to each one of the considered QRTs, which we hope could be used as a starting point for beginners wishing to learn more about each specific framework. In Section~\ref{sec:theory} we present theoretical results where we analytically compute the expected value of the different resource witnesses for special types of families. Finally, we perform numerical studies in  Section~\ref{sec:numerics}  for datasets of free states on $n=3,4,\ldots,8$ qubits where we compute all resource witness for each state in the dataset. This allows us to concurrently cross-examine the different types of resource that a single state can possess, as well as analyze their correlations. Our work finishes with conclusions in Section~\ref{sec:conclusions} where we highlight that our observations are aimed at providing key insights and pointing the community towards future research directions which could lead to more rigorous and theoretical proofs.

\section{Framework}\label{sec:framework}

Let $\HC=\mathbb{C}^d$ be a quantum Hilbert space and $\text{U}(d)$ the unitary group of degree $d$. In what follows, we will define a QRT in terms of two basic ingredients: free operations and free states~\cite{chitambar2019quantum}. The set of free operations $\mathbb{G}\subseteq\text{U}(d)$ is taken to be a unitary representation of a group\footnote{Note that, in general, one can also define the free operations to also contain resource non-increasing channels. However, we will here focus on the case of unitary, and therefore resource-preserving, free operations.}.  Then, the set of pure free states are denoted as $\SC\subseteq \HC$ and can be obtained as the orbit of some reference free state $\ket{\psi^{\rm ref}}$ under $\mathbb{G}$, i.e.,
\begin{equation}\label{eq:free-states-orbit}
    \SC=\{U\ket{\psi^{\rm ref}}\,|\,U\in \mathbb{G}\}\,.
\end{equation}
Equation~\ref{eq:free-states-orbit} shows the fact that $\SC$ can be implicitly defined via some $\ket{\psi^{\rm ref}}$. While there is some freedom in how the reference state is chosen, we will see that when $\mathbb{G}$ is a Lie group there is a simple rule of thumb one can follow. Namely, denoting as $\mathfrak{g}$ the Lie algebra associated to $\mathbb{G}$, then one chooses $\ket{\psi^{\rm ref}}$ as the highest-weight state of $\mathfrak{g}$ (leading to the so-called generalized coherent states~\cite{barnum2003generalizations,barnum2004subsystem,perelomov1977generalized,gilmore1974properties,zhang1990coherent}). While the previous choice will be found to be satisfied for most QRTs considered in this work, we will also present a QRT where $\ket{\psi^{\rm ref}}$ is not the highest-weight state, thus illustrating the freedom that exists in the definition of free states.

Next, given a QRT, one is usually interested in studying and characterizing the resourcefulness of a given (not necessarily free) state. The previous can be accomplished via resource witnesses, i.e., functions $\Lambda:\HC\rightarrow[0,1]$ that are maximized for free states, and whose value decreases with the resourcefulness of the state. Such quantities must satisfy some important properties such as being monotonic (in the sense that smaller values indicate more resource), as well as being $\mathbb{G}$-invariant (i.e, to remain unchanged under free operations). The latter implies that given any $\ket{\psi}\in\HC$ and $U\in\mathbb{G}$, one has
\begin{equation}\label{eq:group-invariant}
\Lambda(U\ket{\psi})=\Lambda(\ket{\psi})\,.
\end{equation}
Crucially, we remark that one can also alternatively define the set of free states as $
\argmin_{\ket{\psi}\in\HC}\Lambda(\ket{\psi})=\SC$, which further shows that different group-invariants (i.e., functions satisfying Eq.~\eqref{eq:group-invariant}) can define distinct QRTs for the same set of free operations. Indeed, the freedom in $\ket{\psi^{\rm ref}}$ can be translated to a freedom in how $\Lambda$ is chosen. 

At this point, we note that several strategies can be employed to define witnesses in a QRT. Indeed, already when studying entanglement there are continuous families of entanglement measures, leading to veritable infinite families of quantifiers~\cite{horodecki2009quantum,renyi1961measures}. Given that our goal is to be able to compare resources across different QRTs, we need to choose witnesses that exist across all theories and that can be somewhat studied in equal footing. Hence. we will consider those that arise from norms, or ``purities'', of group-Fourier decompositions in the irreducible representations induced by $\mathbb{G}$~\cite{barnum2004subsystem,barnum2003generalizations}. Indeed, such approach follows the spirit of signal process analysis, where a function is studied by its fundamental harmonic components; with the added benefit that for quantum states, there exists an purity that is a resource witness.  We refer the reader to ~\cite{bermejo2025characterizing,mele2025clifford,diaz2025unified} for further details on this group-Fourier decomposition program. Hence, all witnesses considered will take the form
\begin{equation}\label{eq:group-invariant-def}
\Lambda(\ket{\psi})=C\sum_{P\in\PC}\bra{\psi}P\ket{\psi}^{2k}\,,
\end{equation}
for some set of Hermitian orthogonal operators $\PC$, $k\in\{1,2\}$,  and where $C$ is a normalization coefficient chosen such that $\Lambda(\ket{\psi})=1$ if $\ket{\psi}\in\SC$.

\begin{figure}[t]
    \centering
    \includegraphics[width=.85\linewidth]{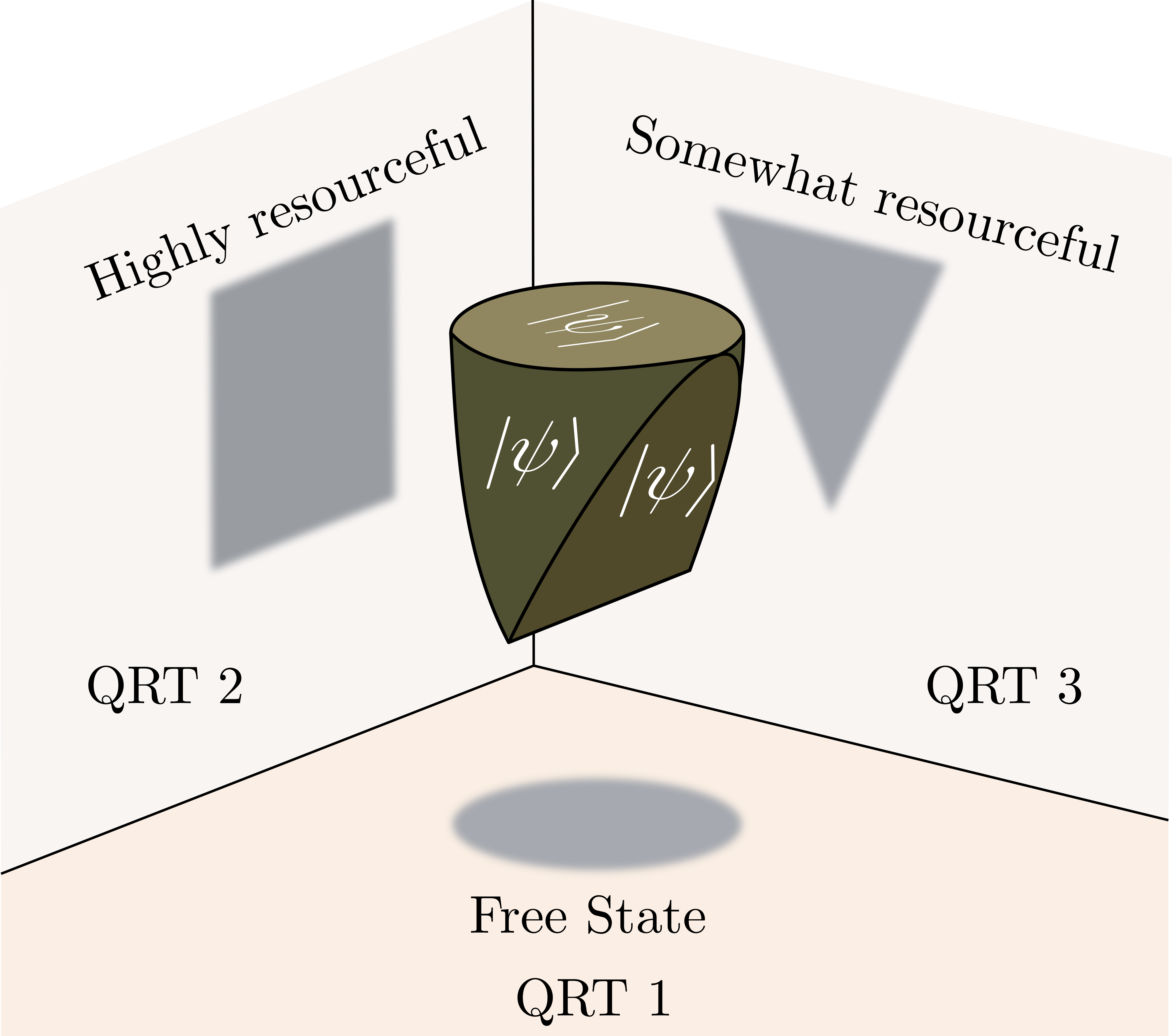}
    \caption{\textbf{Graphic representation of our results.} Our work showcases the fact that the resourcefulness, or the ``quantumness'' of a given quantum state $\ket{\psi}$ is relative to the QRT through which it is analyzed. For instance, a resource-free state of QRT 1, could have widely varying resourcefulness when examined through QRTs 2 and 3.  }
    \label{fig:schematic}
\end{figure}

\section{Considered QRTs}\label{sec:QRTs}

In this section we will present several QRTs to be analyzed throughout the rest of this work. As previously mentioned, our goal is to study how resourceful the free states of one theory are when examined through the optics of another QRT. However, before proceeding to a case-by-case description of QRTs, we find it important to make several remarks that will motivate our studies. First, given two different QRTs, defined from $\{\mathbb{G}_i,\SC_i,\Lambda_i\}$ and $\{\mathbb{G}_j,\SC_j,\Lambda_j\}$, one will have, in general, that the free states of one QRT are not free states of the other. That is, if we take some generic state $\ket{\psi}$  from $ \SC_j$, then
\begin{equation}
    \Lambda_i(\ket{\psi})\leq1\,.
\end{equation}
Specifically, we can expect a strict inequality in the previous equation. Second, and more interestingly, since the free operations of a QRT need not be free in the other one, we find that given some state $\ket{\psi}\in \SC_j$ and unitary $U\in \mathbb{G}_j$, then
\begin{equation}
    \Lambda_i(\ket{\psi})\neq \Lambda_i(U\ket{\psi})\,.
\end{equation}
As such, the set  $\SC_j$ of the $j$-th QRT could simultaneously contain near-free states for the $i$-th QRT, but also highly resourceful ones.

\subsection{Entanglement}

First, we consider the most ubiquitous QRT: multipartite  entanglement~\cite{horodecki2013quantumness,chitambar2019quantum,bennett1996concentrating,vedral1997quantifying,bennett1999quantum}. The theory of entanglement~\cite{horodecki2009quantum} owes its inception to the field of quantum information processing~\cite{wilde2013quantum,plenio1998entanglement}. Indeed, it has been shown that entanglement is a fundamental resource which enables beyond-classical protocols~\cite{horodecki2009quantum,gigena2020one} for quantum communications~\cite{bennett1993teleporting,barrett2002nonsequential,cleve1997substituting,gigena2017bipartite}, and quantum computing~\cite{ekert1998quantum,nielsen2000quantum,datta2005entanglement}. 

In the QRT of multipartite $n$-qubit entanglement the Hilbert space takes the form $\HC=(\mathbb{C}^2)^{\otimes n}$. The set of unitary free operations $\mathbb{G}_{{\rm ent}}$ are given by local unitaries, i.e., by the standard representation of the group $\SU(2)\times\SU(2)\times\cdots\times \SU(2)$. That is, any  $U\in\mathbb{G}_{{\rm ent}}$ can be expressed as $U=U_1\otimes U_2\otimes\cdots\otimes U_n$ for $U_1,U_2,\ldots,U_n\in\SU(2)$. Second, the set of free unentangled states $\SC_{{\rm ent}}$ corresponds to tensor product states obtained by elements of $\mathbb{G}_{{\rm ent}}$ applied to the reference highest-weight state $\ket{\psi^{\rm ref}_{{\rm ent}}}=\ket{0}^{\otimes n}$. Finally,  while several quantifiers of  multipartite entanglement exist~\cite{brennen2003observable,meyer2002global,walter2013entanglement,wong2001potential,carvalho2004decoherence,foulds2020controlled,beckey2021computable,schatzki2022hierarchy}, we will here focus on the quantity
\begin{equation}
    \Lambda_{{\rm ent}}(\ket{\psi})=\frac{1}{n}\sum_{P\in\PC_{{\rm ent}} }\bra{\psi}P\ket{\psi}^2\,,
\end{equation}
where we have defined the set of local Pauli operators $\PC_{{\rm ent}}=\cup_{i=1}^n\{X_i,Y_i,Z_i\}$. One can readily see that $\Lambda_{{\rm ent}}$ is a witness of multipartite entanglement as it is proportional to the purity of the local density matrices~\cite{brennen2003observable,meyer2002global,beckey2021computable}, and is strictly smaller than one as long as any subset of qubits is entangled. That is, defining $\rho_i=\Tr_{\overline{i}}[\dya{\psi}]$ as the marginal on the $i$-th qubit, then  $\Lambda_{{\rm ent}}(\ket{\psi})= \frac{2}{n}(\sum_{i = 1}^n \Tr[\rho_i^2]-\frac{1}{2})$. 

\subsection{Fermionic non-Gaussianity}

Next, we consider the  QRT of fermionic non-Gaussianity~\cite{lami2018guassian, takagi2018convex, zhuang2018resource, chitambar2019quantum,denzler2024learning,bittel2025pac,wan2022matchgate,mele2024efficient,oszmaniec2022fermion}. We recall that computation based on evolving a fermionic Gaussian state via free-fermionic --or matchgate-- unitaries constitutes a  restricted model of quantum computing~\cite{valiant2001quantum,knill2001fermionic,terhal2002classical,divincenzo2005fermionic}. Crucially, while matchgate circuits can be efficiently classically simulable~\cite{valiant2001quantum, knill2001fermionic,terhal2002classical,bravyi2004lagrangian,divincenzo2005fermionic,somma2006efficient,jozsa2008matchgates,brod2011extending,brod2014computational,brod2016efficient,guaita2024representation,oszmaniec2022fermion,helsen2022matchgate,wan2022matchgate,goh2023lie,diaz2023showcasing, mele2024efficient}, the addition of non-Gaussian states, or equivalently of non-matchgate unitaries, can promote this computational paradigm to  universal quantum computation~\cite{lloyd1999quantum, knill2001scheme, bartlett2002universal, menicucci2006universal, jozsa2008matchgates, 
ohliger2010limitations,brod2011extending,brod2014computational, oszmaniec2017universal, zhuang2018resource}. Moreover, the non-Gaussianity resource has been shown to have operational meaning in a variety of quantum information tasks within the related bosonic Gaussian setting, including entanglement distillation~\cite{eisert2002distilling, giedke2002characterization, fiurasek2002guassian, zhang2010distillation, lami2018guassian}, quantum error correction~\cite{gottesman2001encoding, niset2009no}, optimal metrology~\cite{adesso2009optimal}, Bell inequality violation~\cite{banaszek1998nonlocality, banaszek1999testing, filip2002violation, chen2002maximal, garcia2004proposal, nha2004proposed, invernizzi2005effect, garcia2005loophole, ferraro2005nonlocality,  thearle2018violation} and optimal cloning~\cite{cerf2005non}. 

In the QRT of fermionic non-Gaussianity, we take $\HC=(\mathbb{C}^2)^{\otimes n}$. Then, defining the set of $2n$ Majoranas $\{\gamma_i\}_{i=1}^{2n}$ as
\begin{equation} \label{eq:maj}
\begin{split}
    \gamma_1&=X\id\dots \id,\; \gamma_3= ZX\id\dots \id, \;\dots,\; \gamma_{2n-1}=Z\dots Z X\,, \\
        \gam_2&=Y\id\dots\id,\; \gamma_4= ZY\id\dots \id, \; \dots,\;\; \gamma_{2n}\;\;\;=Z\dots Z Y\,
\end{split}
\end{equation}
the free operators $U\in\mathbb{G}_{{\rm ferm}}$ are defined as those which take the form $U=e^{\sum_{i<j}h_{ij}\gamma_i \gamma_j}$ for $h_{ij}\in\mathbb{R}$. These unitaries constitute the spinor representation of $\text{SO}(2n)$. From here, the set $\SC_{{\rm ferm}}$ of fermionic Gaussian states is obtained as the orbit of the highest-weight (vacuum) state $\ket{\psi^{{\rm ref}}_{{\rm ferm}}} = \ket{0}^{\otimes n}$ under $\mathbb{G}_{{\rm ferm}}$. Finally, we quantify the non-Gaussianity via
\begin{equation}
    \Lambda_{{\rm ferm}}(\ket{\psi})=\frac{1}{n}\sum_{P\in\PC_{{\rm ferm}} }\bra{\psi}P\ket{\psi}^2\,,
\end{equation}
where $\PC_{{\rm ferm}}=\{i\gamma_j\gamma_k\}_{1\leq j<k\leq 2n}$. As shown in~\cite{diaz2023showcasing},  $ \Lambda_{{\rm ferm}}$ is a proper measure of non-Gaussianity and genuine fermionic correlations~\cite{gigena2015entanglement, gigena2020one, gigena2021many}, as it corresponds to the $2$-norm of the state's covariance matrix.

\subsection{Imaginarity and realness}

In this section we present two QRTs over $\HC=(\mathbb{C}^{2})^{\otimes n}$ which share the same free operations, but differ on the reference state that defines the set of free states. In both cases, the group of free operations $\mathbb{G}_{\text{O}}$ corresponds to the orthogonal group $\text{O}(2^n)$, i.e.,  real-valued unitaries satisfying $U^T U = UU^T = \id_{2^n}$ where $U^T$ denotes the transpose of $U$ and $\id_{2^n}$ the $2^n\times 2^n$ identity matrix.

\subsubsection{Imaginarity}
Since the free operators $\mathbb{G}_{\text{O}}$ are  real-valued unitaries, this group plays a central role in the QRT of imaginarity. Indeed, the study of quantum mechanics and quantum information relies heavily on the use of complex numbers. As such, imaginarity naturally generates a QRT~\cite{hickey2018quantifying,wu2021resource,wu2021operational,wu2024resource, du2025quantifying} due to the fact that under certain circumstances, non-real quantum states and operations become expensive to implement experimentally \cite{wu2021resource, wu2021operational}, and that various tasks in quantum information utilize imaginarity as a necessary resource \cite{hickey2018quantifying, wu2021resource, wu2021operational, wu2024resource, zhu2021hiding, miyazaki2022imaginarity, sajjan2023imaginary, haug2024pseudorandom, jones2023distinguishability, budiyono2023operational, budiyono2023quantifying, wei2024nonlocal}. At a more fundamental level, several works support the idea that quantum mechanics cannot be fully described without imaginary components \cite{renou2021quantum, li2022testing, chen2022ruling, wu2022experimental, bednorz2022optimal, yao2024proposals}, further motivating the study of the imaginarity QRT.

As such, in a QRT where the imaginary component of a state is considered a resource, one must define the set of free states from a reference vector whose entries are all real-valued. Here, one defines the free states $\SC_{{\rm imag}}$ by simply choosing $\ket{\psi^{{\rm ref}}_{{\rm imag}}}=\ket{0}^{\otimes n}$. Then, we can quantify the resourcefulness of the state with 
\begin{equation}
    \Lambda_{{\rm imag}}(\ket{\psi}) = \frac{1}{2^{n}-1} \sum_{P \in \PC_{{\rm sym}}} \bra{\psi} P\ket{\psi}^2\,,
\end{equation}
where $\PC_{{\rm sym}}=\{P\in\{\id,X,Y,Z\}^{\otimes n}/{\id_{2^n}}\,|\, P=P^T\}$ denotes the set of symmetric Pauli operators composed of an even number of $Y$'s. Here we note that, unlike some of our other resource witnesses, the minimum value here is not $0$, but rather $\frac{2^n - 2}{2(2^n-1)}$. This is due to the fact that no pure quantum state's density matrix can be purely imaginary.

\subsubsection{Realness}

At this point we note that the reference state in the QRT of imaginarity was not the highest-weight state $\ket{+_y}^{\otimes n}$ associated to  $\mathbb{G}_{\text{O}}$, where  $\ket{+_y}=\frac{1}{\sqrt{2}}(\ket{0}+i\ket{1})$ is the eigenstate of the Pauli matrix $Y$. This is due to the fact that $\ket{+_y}^{\otimes n}$ has the maximum amount of imaginarity, and hence would be a poor reference for the QRT of imaginarity. As such, we instead also propose a QRT where the we obtain the free states $\SC_{{\rm real}}$ by applying $\mathbb{G}_{\text{O}}$ to $\ket{\psi^{{\rm ref}}_{{\rm real}}}=\ket{+_y}^{\otimes n}$. Here, the realness of a state is a resource, and we can measure it through the  witness~\cite{bermejo2025characterizing}
\begin{equation}
    \Lambda_{{\rm real}}(\ket{\psi}) = \frac{1}{2^{n - 1}} \sum_{P \in \PC_{{\rm asym}}} \bra{\psi} P\ket{\psi}^2\,,
\end{equation}
where $\PC_{{\rm asym}}=\{P\in\{\id,X,Y,Z\}^{\otimes n}\,|\, P=-P^T\}$ denotes the set of anti-symmetric Pauli operators composed of an odd number of $Y$'s. Clearly, the QRTs of imaginarity and realness are  ``mutually exclusive'', and we can showcase this realization from the fact that their associated witnesses satisfy the property (see the Appendix)
\begin{equation}\label{eq:complement-real-imag}
    \Lambda_{{\rm imag}}(\rho) - 
\frac{\Lambda_{{\rm real}}(\rho)}{2^{1-n} - 2} = 1\,.
\end{equation}

\subsection{Spin coherence}

When solving the quantum harmonic oscillator one recognizes coherent states as the ``most classical'' states of the system, as their position and momentum minimize the Heisenberg uncertainty principle~\cite{cohen2019quantum}. Such an idea has been expanded to other scenarios, leading to the notion of generalized coherent states~\cite{robert2021coherent,perelomov1977generalized,zhang1990coherent}. Here, it has been shown that for systems whose dynamics are described by Lie groups, then the highest weight states of the associated Lie algebra minimize generalized uncertainty relations~\cite{delbourgo1977maximum,barnum2004subsystem} and their orbits have underlying K\"ahler structures~\cite{kostant1982symplectic}, further cementing them as being the most classical states of the system. A prototypical example of generalized coherent states are the so-called spin coherent states~\cite{radcliffe1971some} arising in quantum systems with total angular momentum $s$, whose dynamics are governed by the irreducible representation of $\SU(2)$ acting over 
$\HC=\mathbb{C}^d=\{\ket{s,m}\}_{m=-s}^s$ with $d=2s+1$. Notably, such states can be experimentally prepared (e.g., in nuclear magnetic resonance systems)~\cite{nielsen2000quantum,arecchi1972atomic}, and can be used as a basis for macroscopic quantum information protocols~\cite{byrnes2015macroscopic,pyrkov2014quantum,nielsen2000quantum}. 

In the QRT of spin coherence, the set of free operations $\mathbb{G}_{{\rm coh}}$ are given by the irreducible spin-$s$  representation of $\SU(2)$. Specifically, the free operations are obtained from the exponentiation of the  $d\times d$ generators $S_x$, $S_y$ and $S_z$ whose action is given by
\small
\begin{align}
    \bra{s,m'} S_x \ket{s,m} &= \frac{1}{2}(\delta_{m', m+1} + \delta_{m'+1, m}) \sqrt{s(s+1)-m'm} \nonumber \\
    \bra{s,m'} S_y \ket{s,m} &= \frac{1}{2i}(\delta_{m', m+1} - \delta_{m'+1, m}) \sqrt{s(s+1)-m'm} \nonumber \\
    \bra{s,m'} S_z \ket{s,m} &= \delta_{m', m}m .\nonumber 
\end{align}
\normalsize
Then, free states $\SC_{{\rm coh}}$ are obtained as the orbit of the reference highest-weight spin coherent state $\ket{\psi^{{\rm ref}}_{{\rm coh}}}=\ket{s, s}$. Finally, we measure the resourcefulness using the witness introduced in~\cite{barnum2004subsystem}
\begin{equation}
    \Lambda_{{\rm coh}}(\ket{\psi}) = \frac{1}{s^2} \sum_{P \in \PC_{{\rm coh}}} \bra{\psi} P\ket{\psi}^2\,,
\end{equation}
with $\PC_{{\rm coh}}=\{S_x,S_y,S_z\}$. Such quantity measures whether the state maximizes the angular momentum in the spin representation.

\subsection{Clifford non-stabilizerness}

The QRT of non-stabilizerness revolves around the fact that the evolution of stabilizer states through quantum circuits composed of Clifford gates can be efficiently  simulated classically via the Gottesman-Knill theorem~\cite{gottesman1998heisenbergrepresentation, aaronson2004improved, nest2008classical}.  Here, non-stabilizerness can be understood as the resource which promotes this restricted form of computation to universal quantum computing~\cite{veitch2014resource,howard2017application,chitambar2019quantum, leone2022stabilizer}. Moreover, since many error correction codes have Clifford gates as their native operations~\cite{gottesman1997stabilizer, gottesman1998theory,bravyi2005universal,gottesman2009introduction,howard2017application}, then the QRT of non-stabilizerness is crucial to understanding the requirements for fault-tolerance~\cite{bravyi2005universal, preskill1998fault, bravyi2012magic, reichardt2005quantum} through tasks like magic state distillation and
non-Clifford gate compilation~\cite{campbell2011catalysis,howard2017application,beverland2020lower,seddon2021quantifying}.

We define the free operation $\mathbb{G}_{{\rm stab}}$ of the QRT of non-stabilizerness over $\HC=(\mathbb{C}^2)^{\otimes n}$ as the unitaries from the Clifford group $C_n$. For convenience, we recall that given the Pauli group $\mathbb{P}=\{ { \pm 1, \pm i}\}\times \{\id,X,Y,Z\}^{\otimes n}$, then the Clifford unitaries map elements of the Pauli group to elements of the Pauli group
\begin{equation}
    C_n = \{ U \in \text{U}(2^n)\, |\, U P U^\dagger \in \mathbb{P}, \;\forall\, P \in \mathbb{P}\}\,.
\end{equation}
Now, the free operations $\mathbb{G}_{{\rm stab}}$ form a discrete (rather than continuous) group up to global phases. Then, the free states $\SC_{{\rm stab}}$, are given by stabilizer states, i.e., states such that there exists an abelian subgroup $H \subset \mathbb{P}$ of the Pauli group of size $2^n$ such that $U\ket{\psi} = \ket{\psi}$ for every $U \in H$. We can define such a set from the reference state $\ket{\psi^{{\rm ref}}_{{\rm stab}}}=\ket{0}^{\otimes n}$. While there exist several measures of non-stabilizerness~\cite{howard2017application,bravyi2019simulation,leone2022stabilizer,leone2024stabilizer}, we will here focus on the stabilizer R\'enyi entropy of order two
\begin{equation}
    \Lambda_{{\rm stab}}(\psi) = \frac{1}{2^n-1} \sum_{P \in \mathcal{P}_{{\rm stab}}} \bra{\psi}P \ket{\psi}^{4}\,,
\end{equation}
where $\mathcal{P}_{{\rm stab}}=\{\id,X,Y,Z\}^{\otimes n}/\{\id_{2^n}\}$. Note that unlike previously described witnesses, $\Lambda_{{\rm stab}}(\psi)$ is  expressed as a summation of expectation values to the fourth power. This follows from the fact that the Clifford group forms a $3$-design over $\text{U}(2^n)$ (we refer the reader to~\cite{mele2025clifford} for additional details), i.e., any Clifford invariant expressed as a second order power will be constant for all states.

\subsection{Additional Lie-group-based QRTs}

Up to this point we have motivated and presented QRTs that have been widely studied in the literature. In this section we follow the general recipe for defining Lie group-based QRTs and define two new theories based on $S_n$-equivariant unitaries and local uniform unitaries.

\begin{table*}
    \centering
    \begin{tabular}{|c|c|c|c|c|c|}
      \hline & $\mathbb{E}_{\HC}$ & $\mathbb{E}_{\SC_{\rm imag}}$ & $\mathbb{E}_{\SC_{\rm real}}$ & $\lim_{n\rightarrow \infty}\frac{\mathbb{E}_{\SC_{\rm imag}}}{\mathbb{E}_{\HC}}$ & $\lim_{n\rightarrow \infty}\frac{\mathbb{E}_{\SC_{\rm real}}}{\mathbb{E}_{\HC}}$\\
         \hline
       $\Lambda_{{\rm ent}}$  & $\frac{3}{2^n+1}$ & $\frac{4}{2^n+2}$  & $\frac{3\cdot 2^n-2}{(2^n-1)(2^n+2)}$ & $\frac{4}{3}$ & 1\\\hline
        $\Lambda_{{\rm ferm}}$ & $\frac{2n - 1}{2^n + 1}$ & $\frac{2n}{(2^n+2)}$ & $\frac{2^n(2n-1)-2}{(2^n-1)(2^n+2)}$ &1 & 1 \\\hline
        $\Lambda_{{\rm imag}}$ & $\frac{2^n+1}{2^{n+1}+2}$ & 1 & $\frac{2^n-2}{2(2^n-1)}$ & $\infty$ & 1  \\\hline
        $\Lambda_{{\rm real}}$ & $\frac{2^n-1}{2^n+1}$ & 0 & 1 & 0 & 1 \\\hline
        $\Lambda_{{\rm coh}}$ & $\frac{2^n-1}{(2^n+1)^2}$ & $\frac{4}{3}\frac{2^n-1}{(2^n+1)^2}$  & $\frac{3\cdot 2^{2n}+2^n-2}{3(2^n-1)^2(2^n+2)}$ & $\frac{4}{3}$ & 1  \\\hline
        $\Lambda_{{\rm stab}}$ & $\frac{3 }{3+2^n}$ & $\frac{6 }{6+2^n}$ & $\frac{3 \left(3\cdot 2^n+4^n-2\right)}{\left(2^n-1\right) \left(2^n+1\right) \left(2^n+6\right)}$ & $2$ & $1$ \\\hline
        $\Lambda_{S_n}$ & $\frac{Te_{n+1} - 1}{ 2^{2n}-1}$& $\frac{2\left(\frac{1}{6}(\frac{n}{2}+1)(\frac{n}{2}+2)(2n+3)-1\right)}{(2^n-1)(2^n+2)}$ & $\frac{2^nn(n(n+6)+11)-3(n(n+4)+8)}{6(2^n-1)^2(2^n+2)}$ & 1 & 1 \\\hline
       $\Lambda_{{\rm uent}}$ & $\frac{3}{n(2^n+1)}$ & $\frac{4}{n(2^n+2)}$  & $\frac{3\cdot 2^n-2}{n(2^n-1)(2^n+2)}$ & $\frac{4}{3}$ & 1 \\\hline
    \end{tabular}
    \caption{\textbf{Expected witness values for Haar random states, and for random states in $\SC_{\rm imag}$ and $\SC_{\rm real}$.} Here we show the results of Eqs.~\eqref{eq:exp-1}--\eqref{eq:exp-3} for all QRT witnesses considered. In addition, we also compute the ratios $\frac{\mathbb{E}_{\SC_{\rm imag}}}{\mathbb{E}_{\HC}}$ and $\frac{\mathbb{E}_{\SC_{\rm real}}}{\mathbb{E}_{\HC}}$ in the large $n$ limit, which allows us to determine whether the states in $\SC_{\rm imag}$ or $\SC_{\rm real}$ are more resourceful than a Haar random state. }
    \label{tab:Haar-tab}
\end{table*}

\subsubsection{$S_n$-equivariance}

Consider the $n$-qubit Hilbert space $\HC=(\mathbb{C}^2)^{\otimes n}$. First, let us define the symmetric group, $S_n$, consisting of all possible permutations of a list of size $n$, and its qubit-permuting representation $R$, such that for any $\pi\in S_n$
\begin{equation}
R(\pi)|i_1i_2\cdots i_n\rangle =|i_{\pi^{-1}(1)}i_{\pi^{-1}(2)}\cdots i_{\pi^{-1}(n)}\rangle\,.
\end{equation}
Then, the set of free operations $\mathbb{G}_{{S_n}}$ is given by the $S_n$-equivariant unitaries, i.e., $U\in \mathbb{G}_{{S_n}}$ if $\forall \pi \in S_n, \quad  [R(\pi), U] = 0$. To generate the elements in $\mathbb{G}_{{S_n}}$ we can first find all the linearly independent $S_n$-equivariant Hermitian operators, and then exponentiate them. The latter can be found by twirling all the Paulis $P\in\{\id,X,Y,Z\}^{\otimes n}/\{\id_{2^n}\}$ as $\frac{1}{n!}\sum_{\pi\in S_n}R(\pi)P R\ad(\pi)$, leading to the set $\PC_{{S_n}}$ of $Te_{n+1} - 1$ elements, with $Te_{n}=\frac{1}{6}(n(n + 1)(n + 2))$ the tetrahedral numbers~\cite{kazi2023universality,schatzki2022theoretical}. We further normalize the twirls, such that $\Tr[P^2]=2^n$ for each $P\in\PC_{S_n}$. Next, the free states, $\SC_{{S_n}}$, are obtained as the orbit of the highest-weight state $\ket{\psi^{{\rm ref}}_{S_n}}=\ket{0}^{\otimes n}$ under $\mathbb{G}_{{S_n}}$.  Finally, the resourcefulness witness is
\begin{equation}
    \Lambda_{{S_n}}(\ket{\psi}) = \frac{1}{2^n-1}\sum_{P \in\PC_{{S_n}}} \bra{\psi}P \ket{\psi}^{2}\,.
\end{equation}

Here we note that while the QRT of $S_n$-equivariance has not been formally explored in the literature we can motivate its study from the fact that $S_n$-equivariant circuits constitute a restricted form of computation that can be efficiently classically simulated~\cite{anschuetz2022efficient}. Moreover, these circuits have recently played a central role in quantum machine learning~\cite{schatzki2022theoretical,nguyen2022atheory}, thus illustrating their power when used in conjunction with $S_n$-resourceful states.

\subsubsection{Non-uniform entanglement}

Next, let us introduce a QRT that lies at the intersection of multipartite entanglement and $S_n$-equivariance. Namely, we consider an  $n$-qubit Hilbert space $\HC=(\mathbb{C}^2)^{\otimes n}$, where the free operations $\mathbb{G}_{{\rm uent}}$ arise from the $n$-fold tensor product of $SU(2)$. That is, any  $U\in\mathbb{G}_{{\rm uent}}$ can be expressed as $U=V\otimes V\otimes\cdots\otimes V$, for $V\in\SU(2)$.  The free states $\SC_{{\rm uent}}$ of this QRT are obtained from the reference, highest-weight, state $\ket{\psi^{{\rm ref}}_{{\rm uent}}}=\ket{0}^{\otimes n}$. Clearly, $\mathbb{G}_{{\rm uent}}\subseteq \mathbb{G}_{{\rm ent}}$ and also $\mathbb{G}_{{\rm uent}}\subseteq \mathbb{G}_{{S_n}}$; and similarly $\SC_{{\rm uent}}\subseteq \SC_{{\rm ent}}$ and also $\SC_{{\rm uent}}\subseteq \SC_{{S_n}}$. The previous implies that the free states in $\SC_{{\rm uent}}$ are also free in the QRTs of multipartite entanglement and of $S_n$-equivariance. Here, we can quantify the resourcefulness of a state via the non-uniform entanglement witness
\begin{equation}
    \Lambda_{{\rm uent}}(\ket{\psi})= \frac{1}{n^2}\sum_{P \in \PC_{{\rm uent}}}\bra{\psi}P \ket{\psi}^{2}
\end{equation}
where $\PC_{{\rm uent}} = \{ \sum_{i=1}^nX_i, \sum_{i=1}^nY_i, \sum_{i=1}^nZ_i\}$~\cite{bermejo2025characterizing}.

\section{Theoretical results}\label{sec:theory}

We begin our analysis by presenting theoretical results where we analytically compute the resourcefulness for different families of states, as well interesting bounds between resource witnesses. We refer the reader to the Appendices for a derivation of all our theoretical results. We note that to simplify comparisons between QRTS we will assume that the QRT of spin coherence is evaluated for a spin $s=(2^n-1)/2$ system.

\subsection{Resourcefulness of Haar random $n$-qubit states, and of random states in $\SC_{{\rm imag}}$ and $\SC_{{\rm real}}$}

To begin, we compute the average expected resourcefulness across the different QRTs of a state randomly sampled according to the Haar measure over $\HC$. That is,   
\begin{equation}\label{eq:exp-1}
    \mathbb{E}_{\HC} \left[ \Lambda(\ket{\psi})\right]=\mathbb{E}_{U\sim\text{U}(2^n)} \left[ \Lambda(U\ket{0^{\otimes n}})\right]\,.
\end{equation}
The results are shown in Table~\ref{tab:Haar-tab}. Here we can see that, as expected, Haar random states will be extremely resourceful for all QRTs as the value of the witnesses exponentially converges to their minimum. The only exception is the value of $\Lambda_{{\rm real}}$, which converges to its maximum of one, indicating that Haar random states are expected to be complex and hence have no resource in the QRT of realness. 

Next, we compute the average expected resourcefulness across the different QRTs for a real-valued Haar random state in $\SC_{\rm imag}$ and for a random state in $\SC_{\rm real}$. In particular, we here define
\begin{align}\label{eq:exp-2}
    \mathbb{E}_{\SC_{\rm imag}}[\Lambda(\ket{\psi})]=\mathbb{E}_{U\sim\text{O}(2^n)}[\Lambda(U\ket{0}^{\otimes n}]\,,
\end{align}
and 
\begin{equation}\label{eq:exp-3}
    \mathbb{E}_{\SC_{\rm real}}[\Lambda(\ket{\psi})]=\mathbb{E}_{U\sim\text{O}(2^n)}[\Lambda(U\ket{+_y}^{\otimes n}]\,.
\end{equation}
Again, the results are shown in Table~\ref{tab:Haar-tab}. We can see that similar to Haar random states, the states in $\SC_{\rm imag}$ and $\SC_{\rm real}$ are extremely resourceful across all QRTs (except for the ones for which they are free) as their witness values converge to their minimum exponentially fast. 

Notably, we can also directly compare the expected witness values for a Haar random state and for a state sampled from $\SC_{\rm imag}$ and $\SC_{\rm real}$. For instance, here we can see that in the large-$n$ limit, the states in $\SC_{\rm real}$ have the same resourcefulness as that in Haar random states (i.e., all of the ratios are equal to one). While a similar phenomenon occurs for the states in $\SC_{\rm imag}$ across the QRTs of fermionic Gaussianity and $S_n$-equivariance (i.e., the witnesses converge to the same value), Haar random states are actually expected to be more resourceful than those of $\SC_{\rm imag}$ for the QRTs of entanglement (standard and uniform) and spin coherence by a factor of $\frac{4}{3}$, and by a factor of $2$ in the QRT of Clifford non-stabilizerness in the large-$n$ limit. 

In addition, we can also use the results in Table~\ref{tab:Haar-tab} to showcase finite size effects. For instance, for small $n$  we find that across the QRTs of entanglement (standard and uniform), fermionic Gaussianity and $S_n$-equivariance, the states in $\SC_{\rm real}$ are more resourceful than Haar random states (the exception being the QRT of spin coherence and Clifford stabilizerness where Haar random states have a smaller witness value). However, this finite size effect quickly vanishes, with the states in $\SC_{\rm real}$ and Haar having essentially the same expected resourcefulness for large $n$.

\subsection{Resourcefulness of Haar random tensor-product states}

Next, let us consider the average resourcefulness across the QRTs for Haar random tensor product states in $\SC_{{\rm ent}}$. 

\begin{proposition}\label{prop:haar-sep}
    Let $\ket{\psi}=\bigotimes_{j=1}^n\ket{\psi_j}$ be a tensor product state, where each single qubit state $\ket{\psi_j}$ is a Haar random state over $\HC_j=\mathbb{C}^2$ in $\SC_{\rm ent}$. Then, denoting $\mathbb{E}_{\SC_{{\rm ent}}}=\mathbb{E}_{\HC_1}\cdots \mathbb{E}_{\HC_n}$ we find that, on average,
\begin{align}\label{eq:haar-sep}
    &\mathbb{E}_{\SC_{{\rm ent}}} \left[ \Lambda_{{\rm ferm}}(\ket{\psi})\right] =  \frac{n-1+3^{-n}}{n}\xrightarrow[n\rightarrow \infty]{}1 \nonumber\\
    &\mathbb{E}_{\SC_{{\rm ent}}} \left[ \Lambda_{{\rm imag}}(\ket{\psi})\right] =  \frac{-2\cdot3^n+4^n+6^n}{3^n \cdot 2 \cdot (2^n-1)}\xrightarrow[n\rightarrow \infty]{}\frac{1}{2}\nonumber\\
    &\mathbb{E}_{\SC_{{\rm ent}}} \left[ \Lambda_{{\rm real}}(\ket{\psi})\right] =  1-\left(\frac{2}{3}\right)^n\xrightarrow[n\rightarrow \infty]{}1 \nonumber\\
    &\mathbb{E}_{\SC_{{\rm ent}}} \left[ \Lambda_{{\rm stab}}(\ket{\psi})\right] =\frac{\left(\frac{8}{5}\right)^n-1}{2^n-1}\xrightarrow[n\rightarrow \infty]{}0  \nonumber\\
    &\mathbb{E}_{\SC_{{\rm ent}}} \left[ \Lambda_{S_n}(\ket{\psi})\right] =  \frac{19 \cdot(3^{n}-1)-2n(n+6)}{8\cdot 3^n (2^n-1)}\xrightarrow[n\rightarrow \infty]{}0\nonumber \\
    & \mathbb{E}_{\SC_{{\rm ent}}} \left[ \Lambda_{{\rm uent}}(\ket{\psi})\right] = \frac{1}{n} \xrightarrow[n\rightarrow \infty]{}0\,.
\end{align}
\end{proposition}
We can see from the previous proposition that as the system size increases the resourcefulness of random tensor product states increases and becomes maximal for the QRTs of imaginarity, non-stabilizerness, $S_n$-equivariance and non-uniformity. Then, let us highlight the fact that the expected value of $\mathbb{E}_{\SC_{{\rm ent}}} \left[ \Lambda_{{\rm ferm}}(\ket{\psi})\right]$ converges to its maximum of one as $n$ increases, indicating that random tensor product states essentially become fermionic Gaussian states.

Next, we also find it interesting to study the non-stabilizerness and fermionic non-Gaussianity for  (uniform and non-uniform) tensor product states in $\SC_{{\rm uent}}$ and $\SC_{{\rm ent}}$. In particular, we note that all the states in $\SC_{{\rm uent}}$ can be parametrized--without loss of generality--as $\ket{\psi_{\rm uent}} =(R_z(\alpha)R_y(\beta)\ket{0})^{\otimes n}$, where $R_\mu$ indicates a single-qubit rotation about the $\mu$-th axis. We obtain
\begin{proposition}\label{prop:uent-ferm-ent}
    Given a uniform state $\ket{\psi} =(R_z(\alpha)R_y(\beta)\ket{0})^{\otimes n}$, we find that its fermionic non-Gaussianity witness value is
    \begin{equation}
        \Lambda_{{\rm ferm}}(\ket{\psi}) = \frac{n + \cos^{2n}(2\beta) - 1}{n}.
    \end{equation}
\end{proposition}
From Proposition~\ref{prop:uent-ferm-ent}, we can readily derive the  following two corollaries.
\begin{corollary}\label{cor:uent-fer-1}
        The minimum value of $ \Lambda_{{\rm ferm}}$ of any uniform tensor product state in $\SC_{{\rm uent}}$ is
    \begin{equation}
        \min_{\ket{\psi}\in \SC_{{\rm uent}}}\Lambda_{{\rm ferm}}(\ket{\psi}) = \frac{n-1}{n}.
    \end{equation}
\end{corollary}
\begin{corollary}\label{cor:uent-fer-2}
Let $\ket{\psi}=\ket{\phi}^{\otimes n}$ be a tensor product state, where $\ket{\phi}$ is a single qubit  Haar random state over $\mathbb{C}^2$ in $\SC_{\rm uent}$. Then, denoting $\mathbb{E}_{\SC_{{\rm uent}}}=\mathbb{E}_{\HC}$ we find that, on average,
    \begin{align}
   \mathbb{E}_{\SC_{{\rm uent}}}[\Lambda_{{\rm ferm}}(\ket{\psi})]&=1-\frac{2}{2 n+1}\,.
\end{align}
\end{corollary}

Notably, we can also prove that the minimum value of $ \Lambda_{{\rm ferm}}$ for tensor product (non-uniform) states in $\SC_{{\rm ent}}$ is exactly the same one as that for uniform tensor product states of $\SC_{{\rm uent}}$ (as per Corollary~\ref{cor:uent-fer-1}). That is, 
\begin{proposition}\label{prop:ent-fer-1}
        The minimum value of $ \Lambda_{{\rm ferm}}$ of any tensor product state in $\SC_{{\rm ent}}$ is
    \begin{equation}
        \min_{\ket{\psi}\in \SC_{{\rm ent}}}\Lambda_{{\rm ferm}}(\ket{\psi}) = \frac{n-1}{n}.
    \end{equation}
\end{proposition}
While the two lower bounds in Corollary~\ref{cor:uent-fer-1} and Proposition~\ref{prop:ent-fer-1} match, we can use Corollary~\ref{cor:uent-fer-2} to find that
\begin{equation}
    \mathbb{E}_{\SC_{{\rm uent}}}[\Lambda_{{\rm ferm}}(\ket{\psi})]>\mathbb{E}_{\SC_{{\rm ent}}} \left[ \Lambda_{{\rm ferm}}(\ket{\psi})\right]\quad \forall n\geq 2\,,\nonumber
\end{equation}
indicating that a random uniform tensor product state is closer to being a fermionic Gaussian state than a random non-uniform tensor product state (see also Proposition~\ref{prop:haar-sep}).

Then, let us study the non-stabilizerness of the states in $\SC_{{\rm uent}}$. We find the following proposition.
\begin{proposition}\label{prop:uent-stab}
    Given a uniform state $\ket{\psi} =(R_z(\alpha)R_y(\beta)\ket{0})^{\otimes n}$, we find that its Clifford non-stabilizerness witness value is
    \small
    \begin{equation}
        \Lambda_{{\rm stab}}(\ket{\psi}) = \frac{\left(1 + \cos^4(\beta/2) + \frac{1}{4}\left(3 + \cos(2\alpha)\right)\sin^4(\beta/2)\right)^n-1}{2^n-1}\,.
    \end{equation}
    \normalsize
\end{proposition}
Using the result in Proposition~\ref{prop:uent-stab} we find that the states in $\SC_{{\rm uent}}$ of maximal magic correspond to the angles $\alpha=\frac{\pi}{16}$ and $\beta=\arctan(\sqrt{2-\sqrt{3}})/2$, recovering the well known magic state $\ket{T}^{\otimes n}$ with $\dya{T}=\frac{1}{2}(\id+\frac{1}{\sqrt{3}}(X+Y+Z))$~\cite{bravyi2005universal}.

\subsection{Resourcefulness of Haar random Gaussian states}

The last family of states for which we calculate average expected resource values is the set $\SC_{\rm ferm}$.

\begin{proposition}\label{prop:haar-gauss}
    Let $\ket{\psi} = U\ket{0}^{\otimes n} \in \SC_{\rm ferm}$ be a free fermionic Gaussian state, such that $U$ is sampled randomly from the Haar measure over the spinor representation of $\mathbb{SO}(2n)$. We find that the following expectation values of resource witnesses of $\ket{\psi}$ hold:
\begin{align}\label{eq:haar-gauss}
    &\mathbb{E}_{\SC_{{\rm ferm}}} \left[ \Lambda_{{\rm ent}}(\ket{\psi})\right] = \frac{1}{2n-1}\xrightarrow[n\rightarrow \infty]{}0 \nonumber\\
    &\mathbb{E}_{\SC_{\rm ferm}}[\Lambda_{\rm imag}(\ket{\psi})]=\frac{1}{2^n-1}\sum_{k=2,4,\ldots}^{2n}\frac{\left(\binom{2 n}{k}+\binom{n}{\frac{k}{2}}\right) \binom{n}{\frac{k}{2}}}{2 \binom{2 n}{k}}\nonumber\\&\quad\quad\quad\quad\quad\quad\quad\quad\xrightarrow[n\rightarrow \infty]{}\frac{1}{2} \nonumber\\
    &\mathbb{E}_{\SC_{\rm ferm}}[\Lambda_{\rm real}(\ket{\psi})]=\frac{1}{2^{n-1}}\sum_{k=2,4,\ldots}^{2n}\frac{\left(\binom{2 n}{k}-\binom{n}{\frac{k}{2}}\right) \binom{n}{\frac{k}{2}}}{2 \binom{2 n}{k}}\nonumber\\&\quad\quad\quad\quad\quad\quad\quad\quad\xrightarrow[n\rightarrow \infty]{}1 \nonumber\\
    &\mathbb{E}_{\SC_{{\rm ferm}}} \left[ \Lambda_{{\rm uent}}(\ket{\psi})\right] = \frac{1}{n(2n-1)}\xrightarrow[n\rightarrow \infty]{}0\,.
\end{align}
\end{proposition}

Proposition~\ref{prop:haar-gauss} shows that as $n$ increases, fermionic Gaussian states become entangled, as well as have maximal resource for the imaginarity QRT (i.e., they are not real valued states).

\subsection{Resourcefulness of stabilizer states}

To begin, let us recall that the set of stabilizer states $\SC_{{\rm stab}}$ can be defined as some Clifford unitary  $U\in\mathbb{G}_{{\rm stab}}$ applied to the reference state $\ket{0}^{\otimes n}$. In density matrix formalism, we have that for any $\ket{\psi}\in \SC_{{\rm stab}}$
\begin{equation}
    \dya{\psi}= \prod_{j=1}^n\frac{\id_{2^n}+UZ_jU\ad}{2}\,,\nonumber
\end{equation}
where we used the equality $ \dya{0}^{\otimes n}= \prod_{j=1}^n\frac{\id_{2^n}+Z_j}{2}$. Then, since Clifford operators map Paulis to Paulis (up to a phase), we know that each $UZ_jU\ad$, as well as their products, will be Pauli operators. In turn, the previous implies that $\dya{\psi}$ will be expressed as a sum of $2^n$ Paulis, all with the same coefficient $1/2^n$. This realization allows us to prove the following results for the values that the resourcefulness witness of the different Pauli-based QRTs can, and cannot take, over stabilizer states. 
\begin{proposition} \label{prop:cliff-values}
    Let $\ket{\psi}$ be a stabilizer state from $\SC_{{\rm stab}}$. Then, we find:
    \begin{itemize}
        \item The entanglement witness $\Lambda_{{\rm ent}}(\ket{\psi})$ can only take discrete values $\{ \frac{j}{n} \}_{j = 0}^n$, excluding the value $\frac{n-1}{n}$.
        \item The fermionic non-Gaussianity witness $\Lambda_{{\rm ferm}}(\ket{\psi})$ can only take discrete values $\{ \frac{j}{n} \}_{j = 0}^n$, excluding the value $\frac{n-2}{n}$.
        \item $\ket{\psi}$ is either a free state or a maximally resourceful state with respect to the QRT of imaginarity or realness, i.e., it only takes the minimum or maximum resourcefulness value of $\Lambda_{{\rm imag}}$ or $\Lambda_{{\rm real}}$.
    \end{itemize}
\end{proposition}

\subsection{Uniform entanglement inequality}

To finish, we note that the sets of operators $\PC_{{\rm ent}}$ and $\PC_{{\rm uent}}$ respectively defining the resource witnesses for the QRTs of entanglement and non-uniform entanglement are inherently very related. Both sets are made up of local Pauli operators and thus can be found to satisfy a strict inequality given by the following theorem.
\begin{theorem}\label{theo:uent-ent-bound}
    Given an arbitrary state $\ket{\psi} \in \HC$, the following inequality holds
    \begin{equation}
        \Lambda_{{\rm ent}}(\ket{\psi}) \ge \Lambda_{{\rm uent}}(\ket{\psi})\,.
    \end{equation}
    The equality arises if and only if $\bra{\psi}X_i\ket{\psi}=\bra{\psi}X_j\ket{\psi}$, $\bra{\psi}Y_i\ket{\psi}=\bra{\psi}Y_j\ket{\psi}$ and $\bra{\psi}Z_i\ket{\psi}=\bra{\psi}Z_j\ket{\psi}$ $\forall i,j$.
\end{theorem}
An immediate corollary of Theorem~\ref{theo:uent-ent-bound} is 
\begin{corollary}
    Let $\ket{\psi}\in \SC_{S_n}$ be a free state for the QRT of $S_n$-equivariance. Then, $\Lambda_{{\rm ent}}(\ket{\psi}) = \Lambda_{{\rm uent}}(\ket{\psi})$.
\end{corollary}

To finish, we note that if $\mathbb{E}_{\SC}[\bra{\psi}^{\otimes 2}X_i\otimes X_j \ket{\psi}^{\otimes 2}]=0$, then 
    \begin{equation}
        \mathbb{E}_{\SC}[\Lambda_{{\rm uent}}(\ket{\psi})] =\frac{1}{n} \mathbb{E}_{\SC}[\Lambda_{{\rm ent}}(\ket{\psi})]\,.
    \end{equation}
Such is the case for Haar random states, as well as for the states in $\SC_{{\rm ent}}$, $\SC_{{\rm imag}}$  and $\SC_{{\rm real}}$ (see Table~\ref{tab:Haar-tab} as well as Propositions~\ref{prop:haar-sep} and~\ref{prop:haar-gauss}).

\section{Numerical analysis}\label{sec:numerics}

In this section we will numerically study how resourceful the free states of the eight previously defined QRTs are when examined via the witnesses of the other theories. For this purpose we create a data set consisting of $200$ randomly sampled free states from each QRT (i.e., by evolving the reference state with a unitary randomly and uniformly sampled from the set of free operations; see the Appendix for additional details), as well as $200$ Haar random states from $\HC$ on $n=3,4,\ldots,8$ qubits. We then measure the resourcefulness of each state according to all resource witnesses. Our analysis aims to find (i) the resources available to free states of a given QRT, (ii) correlations between different resources and free states and (iii) trends with system size.

\subsection{Average resourcefulness}

\begin{figure}[t]
    \centering
    \includegraphics[width=1\linewidth]{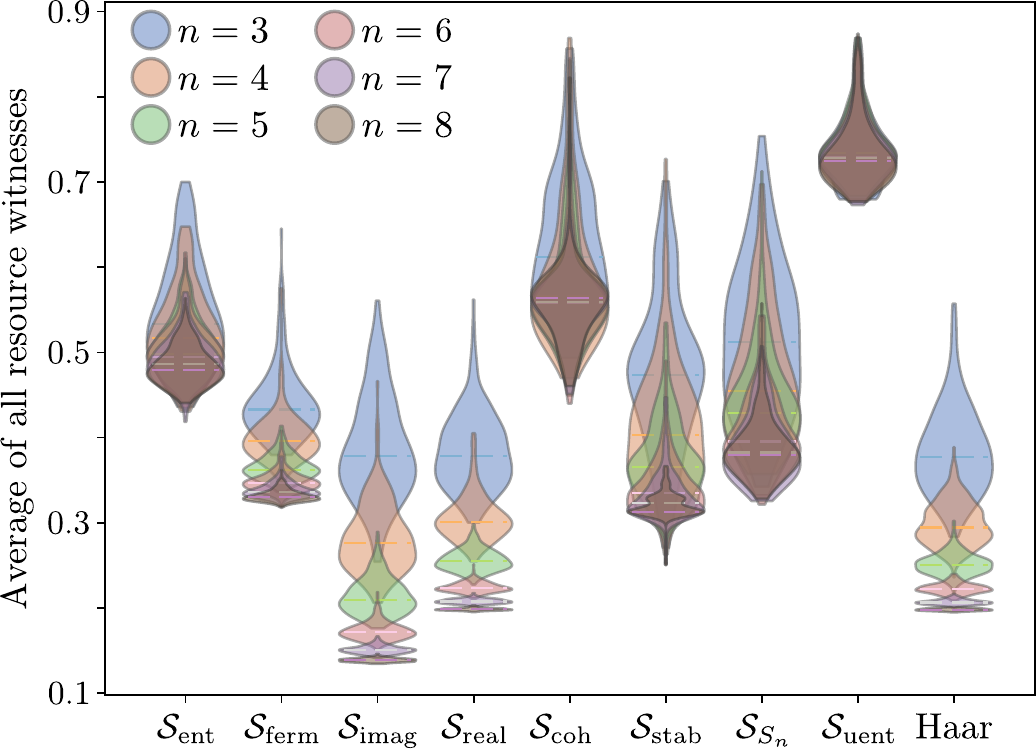}
    \caption{\textbf{Average resourcefulness.} Violin plot of the average resourcefulness across all considered QRTs grouped by each family of the dataset. Dashed lines, represent the median and colors the system size, as indicated.}
    \label{fig:averages}
\end{figure}

Our first coarse-grain analysis is the average resourcefulness across all eight QRTs for all nine families of states in the dataset. 

We present our results in Fig.~\ref{fig:averages} for system sizes of $n=3,4,\ldots,8$ qubits. We first can see that as system size increases, for most families of states the distributions become narrower, indicating that almost all states in the class  possess similar resource values. Moreover, the distributions also tend to shift towards smaller values, meaning that the states are very resourceful for all QRTs (except the ones for which they are defined as free). Notably, the average witness values for the states in $\SC_{\rm ent}$, $\SC_{\rm coh}$ and especially $\SC_{\rm uent}$, remain fairly constant across system sizes in both median and variance. We can understand this behavior from the fact that the states in those families are quite restricted and possess a small number of degrees of freedom. For instance, both the states in  $\SC_{\rm coh}$ and $\SC_{\rm uent}$ have only two degrees of freedom (the two non-trivial Euler angles in the $SU(2)$ unitaries)\footnote{Take for instance a state $\ket{\psi}$ in $\SC_{\rm coh}$. As per Eq.~\eqref{eq:free-states-orbit},  we can express it--without loss of generality--as $\ket{\psi}=U|s,s\rangle$ with $U=e^{i \alpha J_z}e^{i \beta J_y}e^{i \eta J_z}$ and where $\eta$ and $\alpha$ are uniformly sampled according to the standard Haar measure $\sin(\beta)d\alpha d\beta d\eta$. Then, the action of $e^{i \eta J_z}$ over $\ket{s,s}$ only creates an unimportant global phase, showing that $\ket{\psi}$ is only parametrized by $\alpha$ and $\beta$. A similar argument can be given for the states in $\SC_{\rm uent}$.}. 

In addition, we can use Fig.~\ref{fig:averages} to get a notion of how distant the clusters of different families of states look through the optics of the QRT witnesses. In particular, by fixing a given problem size (i.e., a color), we can compare the values of the medians. Here we can see that for $n=3$, the states in $\SC_{{\rm imag}}$, $\SC_{{\rm real}}$ and $\SC_{{\rm Haar}}$ are all clustered together. However,  as the system size increases, the cluster of free states in the imaginarity QRT separates itself from the other two which remain close and tight (as expected from the results in Table~\ref{tab:Haar-tab}). 

While Fig.~\ref{fig:averages} gives us a high-level overview of the resourcefulness in the states (e.g., indicating that the states in $\SC_{\rm uent}$ appear to be less resourceful overall), a finer-grained analysis is needed. In the next section, we proceed to compare average values for pairs of QRTs. 

\subsection{QRT vs QRT relationships}

\begin{figure*}
    \centering
    \includegraphics[width=1\linewidth]{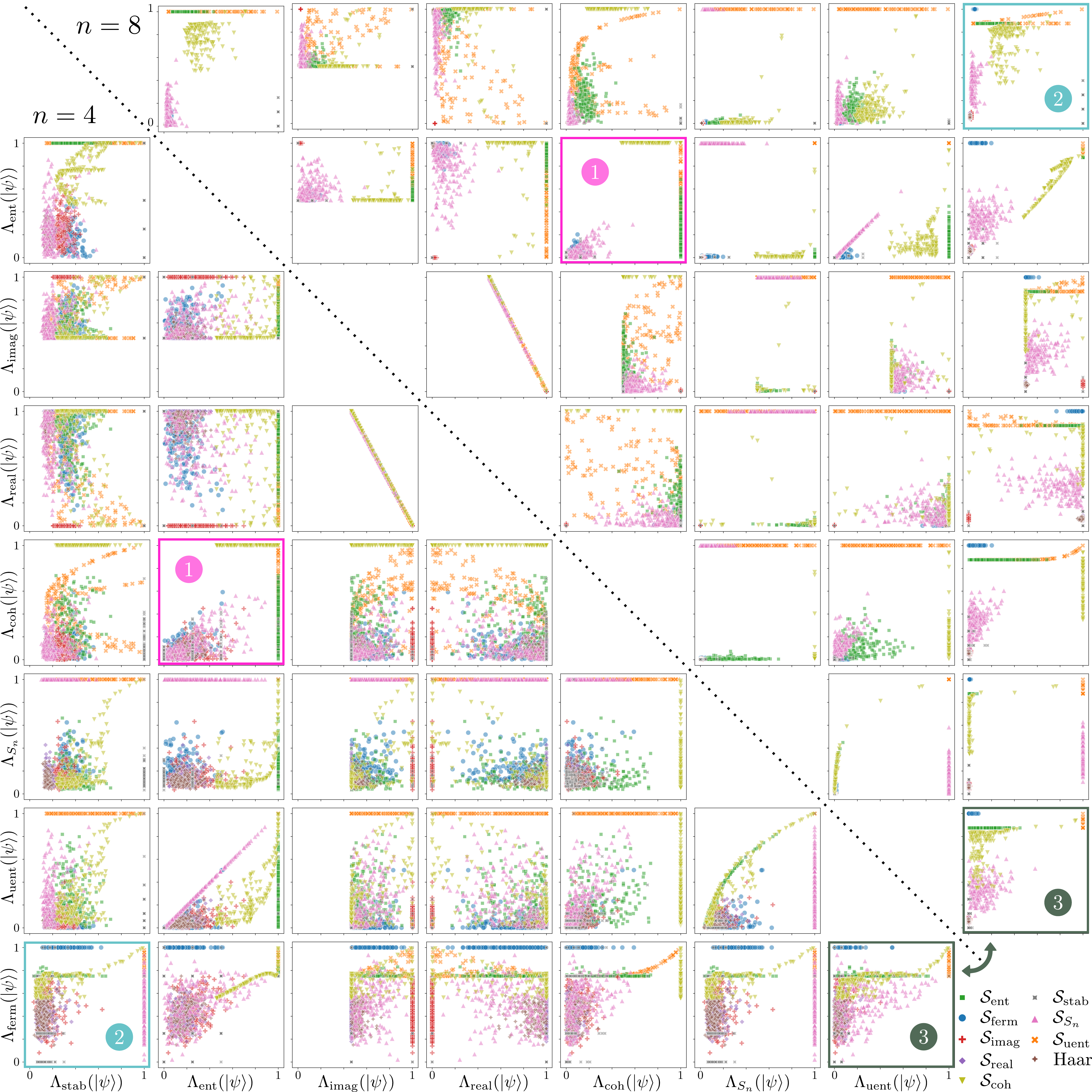}
    \caption{\textbf{Pair plot of QRT vs QRT witnesses.} We computed the witnesses for the states in the datasets on $n=4$ (lower-left half) and $n=8$ qubits (top-right half). These plots serve as a visual aid to look for any relationships between resource witnesses. As indicated by some example figures (with a thick colored edge and a number), the plots are reflected with respect to the diagonal. For instance, one can see how the relation between $\Lambda_{{\rm stab}}(\ket{\psi})$ and $\Lambda_{{\rm ferm}}(\ket{\psi})$ changes with system size by comparing the bottom left plot, and the top-right plot. }
    \label{fig:4-8-qubits}
\end{figure*}

We next turn our analysis by showing in Fig.~\ref{fig:4-8-qubits} a representative grid of pair plots comparing all pairs of resource witnesses for the considered states on $n=4$ and $n=8$ qubits.

Let us discuss how the numerical observations showcase the sanity of our theoretical results.  For example, one can observe that, as predicted by Proposition~\ref{prop:cliff-values}, the stabilizer states in $\SC_{{\rm stab}}$ can only take discrete values for the witnesses of entanglement, fermionic non-Gaussianity, imaginarity, realness and non-uniform entanglement.  Next, we see a negative linear relation in the plot of  $\Lambda_{{\rm real}}$ vs $\Lambda_{{\rm imag}}$ which follows from Eq.~\eqref{eq:complement-real-imag} which states that these QRTs are incompatible: the more real a state is, the less imaginary it can be, and vice-versa. Then, one can also observe a  clear positive trend between $\Lambda_{S_n}$ and $\Lambda_{{\rm uent}}$. This follows from the fact that low resource states in the QRT of non-uniform entanglement are also low resource states in the QRT of $S_n$-equivariance. Hence, as states become more and more uniform their $S_n$-equivariance must also increase. When focusing on the  plot of $\Lambda_{{\rm ent}}$ vs $\Lambda_{{\rm uent}}$, one can see in action the inequality stated in Theorem~\ref{theo:uent-ent-bound} as well as the witness equality for states in $\SC_{S_n}$. In addition, we find that as the system size increases, the value of $\Lambda_{\rm ferm}$ for the states in $\SC_{\rm ent}$ increases, indicating that tensor product states become closer to being Gaussian states as per Proposition~\ref{prop:haar-sep}. Lastly, we note that as per Corollary~\ref{cor:uent-fer-1} we observe a clear lower bound on the value of $\Lambda_{\rm ferm}$ for the states in $\SC_{\rm uent}$ and $\SC_{\rm ent}$.

More generally, we can use Fig.~\ref{fig:4-8-qubits} to get a sense of the distribution and trends of each QRT's free states across all other QRTs. For instance, one can see that many families of states are clustered to specific areas of the plots. Indeed, as system size grows, the clusters tend to become more compact, concentrated towards high-resourcefulness (assuming of course the state is not free for the considered QRT), and the plots sparser (as also showcased in Fig.~\ref{fig:averages}). Crucially, the observation that most clusters shift towards highly-resourceful is a reflection of the fact that as the system size increases, the portion of free states for each QRT becomes negligible as compared to the full Hilbert space dimension. Hence, it is completely expected that a free state of one QRT will tend to be highly resourceful when examined through another QRT's witness.

A finer-grain analysis of the pair plots reveals some clearer trends between resources when we isolate a specific family. As an example, the states in $\SC_{S{\rm coh}}$ on the $\Lambda_{S_n}$ vs $\Lambda_{{\rm uent}}$ plot form a curve that appears to upper bound the witness functional relation. Here, we find that $\SC_{\rm coh}$ states possess low amounts of entanglement and fermionic non-Gaussianity--as measured through $\Lambda_{\rm ent}$ and $\Lambda_{\rm ferm}$ respectively-- and even appear to get less resourceful with system size. Then, one can also see that the fermionic resourcefulness  in $\SC_{S_n}$ does  not appear to converge to maximal resourcefulness as it does for the other QRTs, potentially implying that while $S_n$-equivariant states are non-Gaussian, they are also not fully minimizing $\Lambda_{\rm ferm}$. Next, we find that the wide distribution of $\Lambda_{\rm coh}$ for tensor product states in $\SC_{\rm ent}$ remains stable with system size, as the cluster remains widely spread across all considered values of $n$. Here we also note that as the fermionic non-Gaussianity increases (i.e., $\Lambda_{\rm ferm}$ decreases), the variance of $\Lambda_{\rm ent}$, $\Lambda_{\rm uent}$ and $\Lambda_{\rm coh}$ also decreases. The previous behavior could imply that highly non-Gaussian states must also have large amounts of entanglement, non-uniform entanglement and non-coherence. Then, when examining the pair plot one can see that the states in $\SC_{\rm uent}$ and $\SC_{\rm coh}$ tend to form simple well defined patterns. This is an artifact of the fact that all the states in these datasets can ultimately be parametrized by two Euler angles, leading to a relatively simple family of states (see e.g., the result in Proposition~\ref{prop:uent-ferm-ent}). To finish, we highlight the fact that while the Clifford states only take the values given by our theoretical results, the distribution of these values is not uniform. As expected, the distribution tends to be heavily skewed toward high resource values. We leave the specific analytical study and interpretation of such patterns for future work, but we nevertheless consider it import to highlight them.

\begin{table*}[h!t]
    \text{(a) $n=4$ qubits}\\
    \vspace{.2cm}
    \begin{tabular}{| c|c|c|c|c|c|c|c|c |}
    \hline
        QRT Witness & $\Lambda_{\rm ent}$ & $\Lambda_{\rm ferm}$ & $\Lambda_{\rm imag}$ & $\Lambda_{\rm real}$ & $\Lambda_{\rm coh}$ & $\Lambda_{\rm stab}$ & $\Lambda_{S_n}$ & $\Lambda_{\rm uent}$ \\
    \hline
      $\Lambda_{\rm ent}$   & 1** & \textcolor{brown}{0.275463**} & \textcolor{brown}{0.064775**}   & \textcolor{brown}{-0.064775**} & \textcolor{brown}{0.341529**} & \textcolor{brown}{0.006167}     & \textcolor{brown}{-0.020178}    & \textcolor{brown}{0.512445**}  \\
      $\Lambda_{\rm ferm}$  & 0.419649** & 1** & \textcolor{brown}{0.042085*}    & \textcolor{brown}{-0.042085*}  & \textcolor{brown}{0.181046**} & \textcolor{brown}{0.003349}     & \textcolor{brown}{-0.028957}    & \textcolor{brown}{0.116457**}  \\
      $\Lambda_{\rm imag}$  & 0.065836** & 0.022969 & 1**                        & \textcolor{brown}{-1**}        & \textcolor{brown}{0.043912*}  & \textcolor{brown}{0.000927}     & \textcolor{brown}{0.031568}     & \textcolor{brown}{0.060448**}  \\
      $\Lambda_{\rm real}$  & -0.065836** & -0.022969 & -1**                      & 1**                            & \textcolor{brown}{-0.043912*} & \textcolor{brown}{-0.000927}    & \textcolor{brown}{-0.031568}    & \textcolor{brown}{-0.060448**} \\
      $\Lambda_{\rm coh}$   & 0.612412** & 0.314491** & 0.040090*                 & -0.040090*                     & 1**                           & \textcolor{brown}{-0.024349}    & \textcolor{brown}{0.064772**}   & \textcolor{brown}{0.349989**}  \\
      $\Lambda_{\rm stab}$  & 0.050049** & 0.048205** & -0.044971*                & 0.044971*                      & 0.069263**                    & 1**                             & \textcolor{brown}{-0.039782*}   & \textcolor{brown}{-0.004035}  \\
      $\Lambda_{S_n}$       & 0.313742** & 0.151049** & 0.151844**                & -0.151844**                    & 0.132287**                    & -0.098848**                     & 1**                             & \textcolor{brown}{0.376538**}  \\
      $\Lambda_{\rm uent}$  & 0.734384** & 0.320246** & 0.150656**                & -0.150656**                    & 0.421653**                    & 0.043931*                       & 0.746843**                      & 1**                            \\
    \hline
    \end{tabular}
    
    \vspace{.2cm}\text{(b) $n=8$ qubits}\\
    \vspace{.2cm}
    \begin{tabular}{| c|c|c|c|c|c|c|c|c|}
    \hline
        QRT Witness 
        & $\Lambda_{\rm ent}$ 
        & $\Lambda_{\rm ferm}$ 
        & $\Lambda_{\rm imag}$ 
        & $\Lambda_{\rm real}$ 
        & $\Lambda_{\rm coh}$ 
        & $\Lambda_{\rm stab}$ 
        & $\Lambda_{\rm S_n}$ 
        & $\Lambda_{\rm uent}$  
        \\
    \hline
    $\Lambda_{\rm ent}$   
        & 1** 
        & \textcolor{brown}{0.172454}** 
        & \textcolor{brown}{0.049716}** 
        & \textcolor{brown}{-0.049716}** 
        & \textcolor{brown}{0.268706}** 
        & \textcolor{brown}{0.029725} 
        & \textcolor{brown}{0.027503} 
        & \textcolor{brown}{0.380914}** 
        \\
    $\Lambda_{\rm ferm}$  
        & 0.719580** 
        & 1** 
        & \textcolor{brown}{0.015826} 
        & \textcolor{brown}{-0.015826} 
        & \textcolor{brown}{0.060526}** 
        & \textcolor{brown}{-0.022331} 
        & \textcolor{brown}{0.025334} 
        & \textcolor{brown}{0.058096}** 
        \\
    $\Lambda_{\rm imag}$  
        & -0.103883** 
        & -0.205026** 
        & 1** 
        & \textcolor{brown}{-1}** 
        & \textcolor{brown}{0.015411} 
        & \textcolor{brown}{-0.021483} 
        & \textcolor{brown}{-0.017386} 
        & \textcolor{brown}{-0.024381} 
        \\
    $\Lambda_{\rm real}$  
        & 0.103883** 
        & 0.205026** 
        & -1** 
        & 1** 
        & \textcolor{brown}{-0.015411} 
        & \textcolor{brown}{0.021483} 
        & \textcolor{brown}{0.017386} 
        & \textcolor{brown}{0.024381} 
        \\
    $\Lambda_{\rm coh}$   
        & 0.669931** 
        & 0.481033** 
        & -0.077665** 
        & 0.077665** 
        & 1** 
        & \textcolor{brown}{0.009427} 
        & \textcolor{brown}{-0.018924} 
        & \textcolor{brown}{0.187599}** 
        \\
    $\Lambda_{\rm stab}$  
        & 0.031532 
        & -0.121830** 
        & -0.171714** 
        & 0.171714** 
        & 0.114453** 
        & 1** 
        & \textcolor{brown}{-0.027492} 
        & \textcolor{brown}{0.036516} 
        \\
    $\Lambda_{\rm S_n}$   
        & 0.289907** 
        & 0.224566** 
        & 0.110812** 
        & -0.110812** 
        & 0.082484** 
        & -0.131223** 
        & 1** 
        & \textcolor{brown}{0.138887}** 
        \\
    $\Lambda_{\rm uent}$  
        & 0.656419** 
        & 0.464923** 
        & 0.092123** 
        & -0.092123** 
        & 0.340088** 
        & 0.004385 
        & 0.711510** 
        & 1** 
        \\
    \hline
    \end{tabular}
        \caption{\textbf{Correlation table for the (a) $n=4$ and (b) $n=8$ datasets.} The lower-left (top-right) half is the Pearson correlation coefficient, $r$ value of Eq.~\eqref{eq:pearson-coeff} over the free-state dataset (over a set of 1800 random Haar states; colored brown). Values with a single * indicate results which are statistically significant at the 90\% confidence level. A double ** indicates statistical significance at the 95\% confidence level.}
    \label{tab:compare_reordered-4}
\end{table*}

\begin{table*}[t]
\begin{minipage}{.45\textwidth}
    \text{(a) $n=4$ qubits}\\
    \vspace{.2cm}
    \centering
    \begin{tabular}{| c | c | c | c | c | c | c | c | c |}
    \hline
         & $\Lambda_{\rm ent}$ & $\Lambda_{\rm ferm}$ & $\Lambda_{\rm imag}$ & $\Lambda_{\rm real}$ & $\Lambda_{\rm coh}$ & $\Lambda_{\rm stab}$ & $\Lambda_{\rm S_n}$ & $\Lambda_{\rm uent}$ \\
    \hline
    $\SC_{\rm ent}$   
        & --       & 86.6    & 62.3    & 51.9    & 78.0    & 73.7    & 43.0    & 79.3  \\
    $\SC_{\rm ferm}$ 
        & 34.2    & --       & 67.0    & 47.2    & 45.4    & 70.4    & \textcolor{teal}{77.1}$^{\uparrow}$ & 25.7  \\
    $\SC_{\rm imag}$ 
        & 53.4    & 29.9    & --       & \textcolor{red}{0.5}$^{\downarrow}$  & 38.4    & 53.0    & 50.9    & 41.1  \\
    $\SC_{\rm real}$ 
        & 39.1    & \textcolor{red}{24.8}$^{\downarrow}$  & \textcolor{red}{8.1}$^{\downarrow}$  & --       & 30.4    & 13.8    & 45.2    & 35.8  \\
    $\SC_{\rm coh}$  
        & \textcolor{teal}{98.0}$^{\uparrow}$  & 69.2    & 47.64    & 64.8    & --       & \textcolor{teal}{81.5}$^{\uparrow}$  & 48.9    & 78.2  \\
    $\SC_{\rm stab}$ 
        & \textcolor{red}{30.7}$^{\downarrow}$  & 34.3    & 15.5    & \textcolor{teal}{91.7}$^{\uparrow}$  & \textcolor{red}{24.4}$^{\downarrow}$  & --       & 44.6    & \textcolor{red}{24.2}$^{\downarrow}$  \\
    $\SC_{\rm S_n}$  
        & 53.6    & 40.8    & \textcolor{teal}{73.7}$^{\uparrow}$  & 40.5    & 54.8    & 30.5    & --       & \textcolor{teal}{80.0}$^{\uparrow}$  \\
    $\SC_{\rm uent}$
        & --       & \textcolor{teal}{88.8} $^{\uparrow}$ & 70.7    & 43.5    & \textcolor{teal}{90.8}$^{\uparrow}$  & 66.3    & --       & --     \\
    Haar          
        & 40.6    & 25.2    & 54.6    & 59.5    & 37.4    & \textcolor{red}{10.3}$^{\downarrow}$  & \textcolor{red}{40.0}$^{\downarrow}$  & 35.4  \\
    \hline
    \end{tabular}
        \end{minipage}
    \begin{minipage}{.49\textwidth}
    \text{(b) $n=8$ qubits}\\
    \vspace{.2cm}
    \begin{tabular}{| c | c | c | c | c | c | c | c | c |}
    \hline
         & $\Lambda_{\rm ent}$ & $\Lambda_{\rm ferm}$ & $\Lambda_{\rm imag}$ & $\Lambda_{\rm real}$ & $\Lambda_{\rm coh}$ & $\Lambda_{\rm stab}$ & $\Lambda_{\rm S_n}$ & $\Lambda_{\rm uent}$ \\
    \hline
    $\SC_{\rm ent}$   
        & --       & 90.1    & 62.5    & 51.7    & 86.6    & 80.6    & 70.3    & 84.9  \\
    $\SC_{\rm ferm}$ 
        & 69.9    & --       & 63.2    & 50.9    & 58.9    & 54.8    & \textcolor{teal}{85.5}$^{\uparrow}$ & 45.6  \\
    $\SC_{\rm imag}$ 
        & 40.2    & 25.6    & --       & \textcolor{red}{0.1}$^{\downarrow}$  & 29.3    & 30.1    & 36.2    & 33.3  \\
    $\SC_{\rm real}$ 
        & 28.4    & 21.6  & \textcolor{red}{12.4}$^{\downarrow}$  & --       & 28.9    & 7.7     & \textcolor{red}{26.4}$^{\downarrow}$  & 29.6  \\
    $\SC_{\rm coh}$  
        & \textcolor{teal}{99.9}$^{\uparrow}$  & 71.6    & 26.6    & 82.4    & --       & \textcolor{teal}{96.1}$^{\uparrow}$  & 75.8    & 85.1 \\
    $\SC_{\rm stab}$ 
        & \textcolor{red}{6.4}$^{\downarrow}$   & \textcolor{red}{16.1}$^{\downarrow}$    & 13.8    & \textcolor{teal}{93.3}$^{\uparrow}$  & $\textcolor{red}{7.2}^{\downarrow}$     & --       & 26.4    & \textcolor{red}{4.4}$^{\downarrow}$   \\
    $\SC_{\rm S_n}$  
        & 76.0    & 57.2    & \textcolor{teal}{89.3}$^{\uparrow}$ & 24.9    & 66.7    & 45.2    & --       & \textcolor{teal}{85.6}$^{\uparrow}$ \\
    $\SC_{\rm uent}$
        & --       & \textcolor{teal}{94.9}$^{\uparrow}$ & 76.0    & 37.8    & \textcolor{teal}{95.3}$^{\uparrow}$ & 78.5    & --       & --     \\
    Haar          
        & 28.7    & 22.5    & 55.7    & 58.5    & 26.6 & \textcolor{red}{6.5}$^{\downarrow}$  & 29.1    & 31.1  \\
    \hline
    \end{tabular}
    \end{minipage}
    \caption{\textbf{Pairwise comparison results for the (a) $n=4$ and (b) $n=8$ datasets.} The columns represent the target QRT and the rows represent the wins of each family with respect to the target QRT. The target QRT's free states are excluded. Highlighted in red with a down-pointing arrow (teal with an up-pointing arrow) are the overall losing (winning) family of states, which can therefore be considered the most (least) resourceful according to the target QRT.}
    \label{tab:compare_reordered}
\end{table*}

To supplement our analysis, we present in Table~\ref{tab:compare_reordered-4} the correlation for the eight witnesses considered. The lower left-half presents correlations for the states in the original aforementioned dataset, while the top-right half (colored data) was obtained from an additional dataset of equal size (i.e., $1800$ states) consisting of only Haar random states.  We include this additional dataset in our analysis to detect how much the free-state dataset correlations deviate from those arising from full Haar random states. For completeness, we recall that the Pearson correlation coefficient, $r$, between two variables for a given set of data points $\{(x_i,y_i)\}_{i=1}^N$ is given by

\footnotesize
\begin{equation}\label{eq:pearson-coeff}
    r = 
 \sum_{i} \left(x_i - \bar{x}\right)\left(y_i - \bar{y}\right)/
      \left(\left(\sum_{i}(x_i - \bar{x})^2\right)\left(\sum_{i}(y_i - \bar{y})^2\right)\right)^{1/2},
\end{equation}
\normalsize
where $\bar{x}=\frac{1}{N}\sum_{i=1}^Nx_i$, $\bar{y}=\frac{1}{N}\sum_{i=1}^Ny_i$, is a quantity that measures linear correlation between two sets of data. Therefore, a  value of $r=1$ indicates a perfect positive linear relation between $x$ and $y$; and a perfect negative linear relation for $r=-1$.

As we can see from the tables, while most correlations are weak, they are mostly statistically significant, indicating that our dataset is of appropriate size. The only exceptions among the free state dataset are between $\Lambda_{\rm stab}$  and $\Lambda_{\rm ent}$ or $\Lambda_{\rm uent}$. Notably, for the free-state dataset (lower-left sections of the tables) there are some large correlations among the witnesses  which persist when going from $n=4$ to $n=8$ qubits, and which are larger than those arising for Haar random states. This seems to suggest that certain free states induce some correlations in the QRT witnesses that random states do not reproduce. For instance, one can find relatively large positive correlations between $\Lambda_{\rm ent}$, $\Lambda_{\rm ferm}$, $\Lambda_{\rm coh}$ and $\Lambda_{\rm uent}$. When examining the Haar data, these correlations become much weaker and in fact for $n=8$ the correlation between $\Lambda_{\rm ferm}$ and $\Lambda_{\rm uent}$ is actually negative. Our results also reveal that  some correlations seem to increase or decrease with system size. For example the correlation between $\Lambda_{\rm ent}$ and $\Lambda_{\rm uent}$ is decreasing with $n$ whereas the correlation between $\Lambda_{\rm ent}$/$\Lambda_{\rm uent}$ and $\Lambda_{\rm ferm}$ is increasing. This last result, is possibly explained by our findings that tensor product states become closer to being Gaussian as the system size increases.

\subsection{Pairwise comparison}

In this section, we aim to determine which family of states in the dataset is the most and least resourceful for each QRT (of course, not counting the QRT's free states themselves). While one could draw conclusions from plots such as those in Fig.~\ref{fig:4-8-qubits}, the scattered data and existence of outliers makes it hard to obtain meaningful conclusions without further analysis. As such, we opt to obtain a definitive ranking via  standard pairwise comparison technique. Such analysis, will allow us to compare the resourcefulness of the  free state of all considered QRTs with that of Haar random states, potentially revealing if the naive assumption ``Haar random states should be the most resourceful across all metrics'' holds.

 In our context, a pairwise comparison involves first choosing a resourcefulness witness (for example, from $\Lambda_1$), then selecting two states from different families (which are not free for $\Lambda_1$)--say $\ket{\psi}\in\mathcal{S}_2$ and $\ket{\phi}\in\mathcal{S}_3$--and comparing their witness values. If $
\Lambda_1(\ket{\psi}) > \Lambda_1(\ket{\phi})$, QRT 2 scores one point and QRT 3 scores zero. In the event of a tie (i.e., $\Lambda_1(\ket{\psi}) = \Lambda_1(\ket{\phi}$ up to numerical precision), each QRT scores $\tfrac{1}{2}$ point. We repeat this procedure for every possible pair of states from different families, yielding up to $320,000$ total comparisons per witness. Finally, we convert each QRT’s total score into a win percentage out of the total number of matchups. A higher win percentage indicates that a family's states are closer to being free with respect to the chosen QRT; whereas the family with the lowest win percentage will contain the most resourceful states.

The results of the pairwise comparison are shown in Table~\ref{tab:compare_reordered}. The first notable result is that Haar random states are not the most resourceful states across all QRTs, only having the lowest win percentages for the QRTs of non-stabilizerness and $S_n$-equivariance ($n=4$) or spin coherence and non-stabilizerness ($n=8$); and consistently possessing the most non-stabilizer magic for both system sizes considered. Still, there is a trend where the win percentage of Haar random states decreases as the system size increases, potentially indicating that as $n$ grows Haar random states become more resourceful faster than the states in the other free-state families. Next, we see that stabilizer states tend to be quite resourceful. Indeed, they are the most entangled (and non-uniform entangled) states out of all considered families, with the $n=8$ showcasing that states in $\SC_{\rm stab}$ lose the pairwise comparisons by a landslide\footnote{This result further emphasizes the well known fact that entanglement is not the defining quality of non-classical quantum computation, as stabilizer states evolving through Clifford circuits can be efficiently simulated via the Gottesman-Knill theorem~\cite{aaronson2004improved}, but they also can exhibit high levels of multipartite entanglement.}. On the other hand, we find that tensor product states in $\SC_{\rm ent}$ and spin coherent states from $\SC_{\rm coh}$ typically have little resourcefulness across the considered QRTs, with spin coherent states actually being the less entangled of the lot.

\subsection{Principal component analysis}

\begin{figure}
    \centering
    \includegraphics[width=1\linewidth]{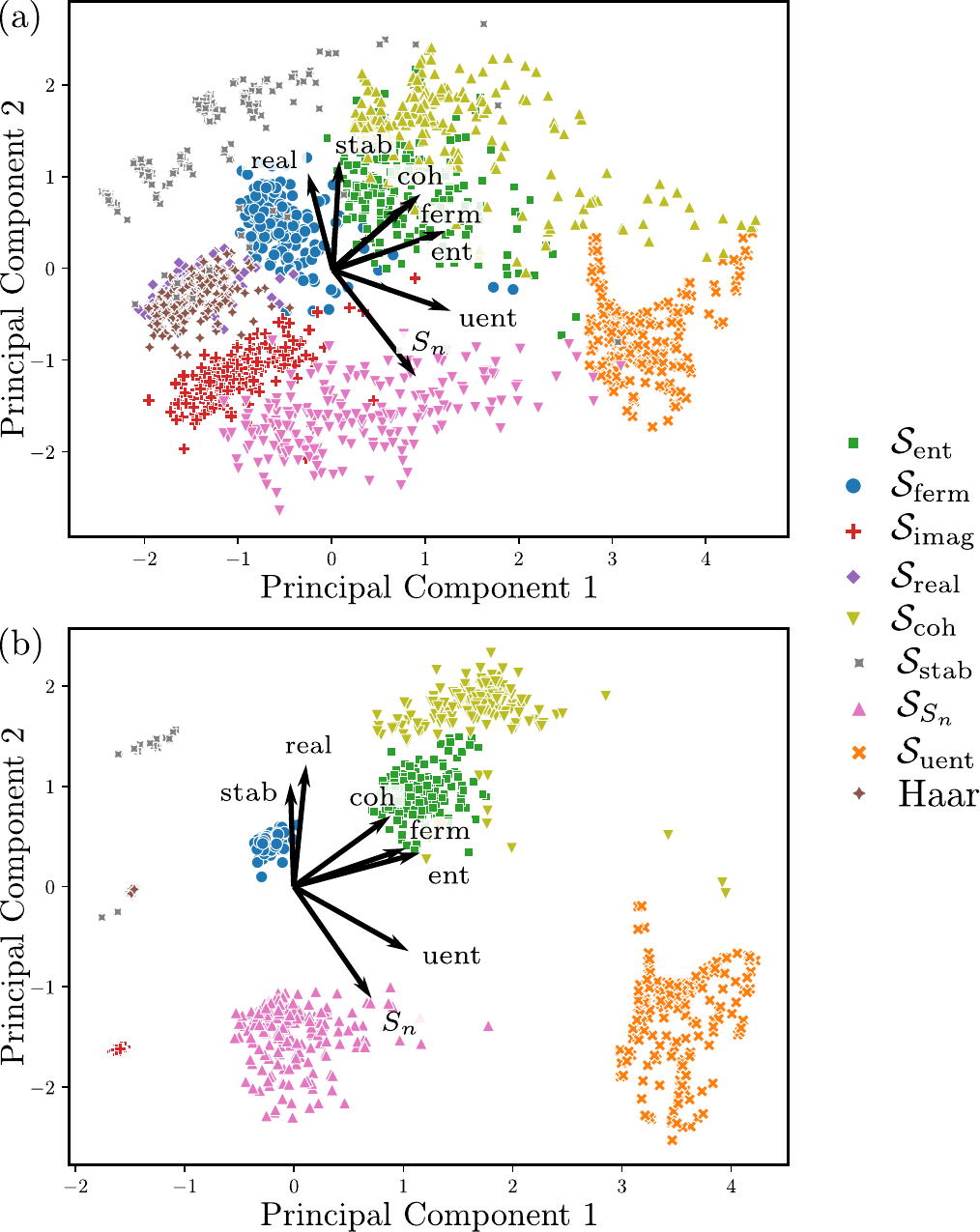}
    \caption{\textbf{Principal component plot for the (a) $n=4$ and (b) $n=8$ datasets.} We perform PCA analysis on the dataset (excluding $\Lambda_{\rm imag}$), allowing us to project each data vector $\vec{v}(\ket{\psi})$ onto the subspace spanned by the first two principal components. The black arrows represent the loading vectors for each QRT's witness.}
    \label{fig:pca}
\end{figure}

Up to this point our analysis has mainly focused on low-level statistics and pairwise comparisons between witnesses and different families of states. In this section we take a more holistic view and instead focus on all witness values for a given state. That is, for all states $\ket{\psi}$ in the dataset we analyze the information in the eight-dimensional vectors $\vec{v}(\ket{\psi})=(\Lambda_{\rm ent}(\ket{\psi}),\Lambda_{\rm ferm}(\ket{\psi}),\cdots,\Lambda_{\rm uent}(\ket{\psi}))$ through a principal component analysis (PCA) in $\mathbb{R}^8$. 

Here we recall that PCA is a widely used technique for visualizing and analyzing the structure of higp
h dimensional data. The procedure finds a linear combination of the variables (the eight resource witnesses in our case) which maximizes the variance in the dataset. This is achieved by taking the covariance matrix of the data and diagonalizing it. From there we can infer global properties of the dataset in the high dimensional space, based on the projection onto the subspace spanned by the principal component vectors. 

To begin, we report the portion of variance explained by each principal component for each dataset. For  $n=4$ we find 
\begin{equation}
    \begin{bmatrix}
    0.39 &
    0.17 &
    0.13 &
    0.12 &
    0.10 &
    0.05 &
    0.01
    \end{bmatrix}\,,
\end{equation}
whereas for $n=8$ we obtain
\begin{equation}
    \begin{bmatrix}
    0.41 &
    0.20 &
    0.14 &
    0.12 &
    0.06 &
    0.03 &
    0.01
    \end{bmatrix}\,.
\end{equation}
Notice that there are only seven entries whereas the vectors $\vec{v}(\ket{\psi})$ are eight dimensional. This is due to the fact that we excluded $\Lambda_{\rm imag}$ from the dataset due to its co-linearity with $\Lambda_{\rm real}$ (equivalently, we could have also excluded $\Lambda_{\rm real}$ and kept $\Lambda_{\rm imag}$). According to these vectors the first two principal components in the $n=4$ ($n=8$) case account for around $57\%$ ($62\%$) of the variance in the data. Thus there is still a large amount of information contained in the last five components.

Graphically, we present the projection of each vector $\vec{v}(\ket{\psi})$ onto the first two principal components in Fig.~\ref{fig:pca}. Here, we also depict the (rescaled) loading vectors for each considered QRT's resource witness. Notably, all loading vectors have a similar magnitude, indicating that the first two principal components have somewhat equal weights among the witnesses, and showing that the data is fairly isotropic. This essentially means low multicolinearity and that the variance accounted for by the first two principal components is evenly spread among all variables. Additionally, the witness vectors seem to be evenly spaced and not opposite one another, entailing that the resource witnesses are not  strongly correlated, at least within the first two principal components. The previous two statements are not quite as supported in the $n=8$ case and in fact we can see instances where loadings are pointing in similar directions. For example the ``${\rm ferm}$'' and ``${\rm ent}$'' loading vectors are quite similar, possibly explained by the fact that $\SC_{\rm ent}$ states converge to Gaussian with system size. One might then expect for the ``uent'' loading to also support this observation, however high uniformity is generally hard to achieve for all families of states, and thus the direction of $\Lambda_{\rm uent}$ is naturally quite distinct from the rest. Similarly the loadings for ``stab'' and ``real'' seem to converge on one another, possibly explained by the trend of $\SC_{\rm stab}$ states to have high $\Lambda_{\rm real}$ resourcefulness as evidence of Table~\ref{tab:compare_reordered}.

Next, one can readily observe that each QRT's vector points toward the respective set of free states, which follows from the definition of resource witness being maximized for the associated free states. However, the same does not occur with the direction associated with $\Lambda_{\rm ferm}$--within the first two principal components. Counterintuitively, the arrow labeled ``ferm'' points away from the projection onto the first two principal components of the states in $\SC_{\rm ferm}$. Such apparent discrepancy can be explained given that there exists many non-Gaussian states (e.g., those in $\SC_{\rm uent}$ and $\SC_{\rm ent}$) whose value of $\Lambda_{\rm ferm}$ can be almost maximal (see Fig.~\ref{fig:4-8-qubits}); a phenomenon not arising for other QRT's witnesses. Thus, it is natural that the direction of the ``ferm'' vector does not point towards the cluster of $\SC_{\rm ferm}$  states. By considering higher order principal components, the direction of $\Lambda_{\rm ferm}$ points to the $\SC_{\rm ferm}$ states (this reflects the incomplete variance accounted for by the first two principal components).

Finally, we find it interesting to note that while PCA is not typically used for clustering (in the machine learning sense), the resulting plots in Fig.~\ref{fig:4-8-qubits} clearly show the data as being highly clustered and increasing in separation with system size. The plots can thus be used to measure how compact these clusters are in general. Specifically, the principal component projection plot for the $n=8$ data offers valuable insight into the landscape of the data. Indeed we can see that $\SC_{\rm uent}$, $\SC_{\rm coh}$ and $\SC_{\rm S_n}$ states seem to have the highest variance whereas $\SC_{\rm real}$, $\SC_{\rm imag}$ and Haar random states are quite compact, supporting similar claims made by observation of Fig.~\ref{fig:averages}.  Furthermore, we can assess which clusters are ``close'' with respect to the first two principal components. For example, the Haar random and the $\SC_{\rm real}$ clusters almost completely overlap and become indistinguishable. On the other hand, the $\SC_{\rm uent}$ and $\SC_{\rm imag}$ clusters are quite isolated from the rest of the data (see also Fig.~\ref{fig:averages}). The fact that classes are generally spread out indicates that the first two principal components are capturing meaningful variance. Additionally, we can use these plots to better visualize outliers within the data and within each family of state. In Fig.~\ref{fig:pca}(b) one can identify two points in $\SC_{\rm stab}$ which stray far from the center of the cluster and seem closer to $\SC_{\rm real}$ and Haar. Similarly $\SC_{\rm coh}$ has several outliers which merge into the clusters of $\SC_{\rm ent}$ and $\SC_{\rm uent}$.

\section{Conclusion}\label{sec:conclusions}

In this work we performed an initial comprehensive analysis on the resourcefulness of the free states in--through the witnesses of--the QRTs of multipartite entanglement, fermionic non-Gaussianity, imaginarity, realness, spin coherence, Clifford non-stabilizerness, $S_n$-equivariance and non-uniform entanglement. Our sprawling collection of theoretical and numerical findings further highlights the rich and extremely complex behavior of the resources that a given state can possess. Indeed, we find that even the simplest of states, e.g., stabilizer, tensor product or Gaussian, can be extremely resourceful through the optics of other theories, therefore cementing the realization that resource is in the eye of the (QRT) beholder. 

By scaling our numerical analysis we can identify certain patterns which either vanish, or that appear to persist with problem size. While our goal was not to rigorously determine if these observations are a reflection of a deeper theoretical result, or mere finite-size effects, we nevertheless believe that their analysis could lead to further insights regarding QRTs. 

To finish, we highlight the fact that this work has only  scratched the surface of a cross-examination of a state's resourcefulness. Namely, we have focused on the simplest purity-like witnesses (e.g., most of which are polynomial of order two in the matrix entries of the state's density matrix). Different witnesses can capture different properties of the same resource, and we expect that performing a similar analysis to the one presented in this work, but using different witnesses, could very well lead to completely different patterns in the results. Similarly, we have only considered resources in pure states. While mixed state monotones are known for entanglement~\cite{horodecki2009quantum}, and some have been defined for magic~\cite{leone2024stabilizer}, it is generally not obvious how to define resource non-increasing channels for general QRTs. Indeed, already for the theory of magic there are both axiomatic and operational definition to non-unitary free operations, and it is has been reported that these  do not match~\cite{heimendahl2022axiomatic}. Given that there is still ambiguity in the definition of free (quantum channel) operations, and in how to properly define monotones (see~\cite{diaz2025unified} for a proposal for a unification), we leave this research direction open to the community.

\section{Acknowledgments}
AED and MC were supported by the U.S. Department of Energy (DOE), Office of Science, National Quantum Information Science Research Centers, Quantum Science Center. P. Braccia and P. Bermejo were supported by the Laboratory Directed Research and Development (LDRD) program of Los Alamos National Laboratory (LANL) under project numbers 20230527ECR and 20230049DR. P. Bermejo also acknowledges constant support from DIPC.  NLD was supported by the Center for Nonlinear Studies at LANL. AAM  acknowledges support by  U.S. DOE through a quantum computing program and a rapid response sponsored by the LANL Information Science \& Technology Institute (ISTI), and by the German Federal Ministry for Education and Research (BMBF) under the project FermiQP. MC was  supported by LANL ASC Beyond Moore’s Law project.   Research presented in this article was supported by the National Security Education Center ISTI using the LDRD program of LANL project number 20240479CR-IST.

\bibliography{quantum}

\clearpage
\newpage
\appendix
\onecolumngrid
\section*{APPENDICES FOR ``\textit{Analyzing the free states of one quantum resource theory as resource states of another}''}

\section{Weingarten calculus}\label{app:weingarten}

In this section we present a brief review of the Weingarten calculus. For a more detailed  description, we refer the reader to~\cite{mele2023introduction}. In particular, given a group $G\subseteq\text{U}(2^n)$, we are interested in computing the $t$-th fold twirl of an operator $X\in\LC(\HC^{\otimes t})$, where $\LC(\HC^{\otimes t})$ denotes the space of linear operators acting on $t$ copies of the Hilbert space. Namely, we want to evaluate the quantity
\begin{equation}
    \tau^{(t)}_G[X]=\int_Gd\mu(U) U^{\otimes t} X (U\ad)^{\otimes t}\,,
\end{equation}
where $d\mu$ denotes the Haar measures over $G$. Crucially, it is well known that the twirl is a projector on the the $t$-th fold commutant of $G$, denoted as ${\rm comm}^{(t)}(G)$, given by 
\begin{equation}
    {\rm comm}^{(t)}(G)=\{A\in \LC(\HC^{\otimes t})\,\,|\,\, [A,U^{\otimes t}]=0\,,\quad \forall U\in G\}.
\end{equation}
Using the previous fact, we can then explicitly evaluate the $t$-th fold twirl as~\cite{mele2023introduction}
\begin{equation}\label{eq:weingarten}
    \tau^{(t)}_G[X]=\sum_{\mu,\nu=1}^{\dim({\rm comm}^{(t)}(G)) } (W^{-1})_{\mu \nu}\Tr[M_\mu\ad X]M_\nu \,,
\end{equation}
where $\{B_\mu\}_{\mu=1}^{\dim({\rm comm}^{(t)}(G)) }$ forms a basis for ${\rm comm}^{(t)}(G)$ and $W$ is the associated Gram matrix with entries $(W)_{\mu \nu}=\Tr[M_\mu\ad M_\nu]$.

\section{Derivation of our theoretical results}

We here present the proof of our theoretical results.

\subsection{Expected witness values for Haar random states, proof of results in Table~\ref{tab:Haar-tab}}

Here we provide proofs for the results in Table~\ref{tab:Haar-tab} for Haar random states. To begin, we recall that a Haar random state $\ket{\psi_H}$ over the $2^n$-dimensional Hilbert space $\HC$ can be obtained by evolving the reference state  $U\ket{0}^{\otimes n}$ with a unitary $U$ sampled according to the Haar measure over $\text{U}(2^n)$. As such, we henceforth define expectation values of Haar random states as $\mathbb{E}_{\HC}[f(\dya{\psi})]=\mathbb{E}_{U\sim\text{U}(2^n)}[f(U\dya{0}^{\otimes n}U\ad)]$. The previous allows us to map the problem of computing expectation values over states to that of computing expectation values over unitaries, and therefore using Eq.~\eqref{eq:weingarten}. In particular, one  finds that~\cite{harrow2013church} 
\begin{equation}\label{eq:Haar}
    \mathbb{E}_{\HC}[\dya{\psi_H}^{\otimes 2}]=\mathbb{E}_{\HC}[U^{\otimes 2}\dya{0}^{\otimes 2n}(U\ad)^{\otimes 2}]=\frac{\id\otimes \id +{\rm SWAP}}{2^n(2^n+1)}\,,
\end{equation}
where ${\rm SWAP}=\sum_{i,j=1}^{2^n}\ket{ij}\bra{ji}$ is the operator that swaps the two copies of the Hilbert space $\HC$. Above, we used the fact that  ${\rm comm}^{(t)}(\text{U}(2^n))={\rm span}_{\mathbb{C}}\{R(S_t)\}$, where $R$ denotes the system permuting representation of the Symmetric group $S_t$. That is
\begin{equation}\label{eq:prem}
R(\sigma \in S_t)|i_1\cdots i_t\rangle=    |i_{\sigma^{-1}(1)}\cdots i_{\sigma^{-1}(t)}\rangle\,,
\end{equation}
and hence ${\rm comm}^{(2)}(\text{U}(2^n))={\rm span}_{\mathbb{C}}\{\id\otimes \id,\SWAP\}$.

Equipped with the previous result, we can readily find that, for QRTs whose resourcefulness witnesses can be expressed as
\begin{equation}
    \Lambda_{\rm qrt}(\ket{\psi})=C\sum_{P\in\PC} \Tr[\dya{\psi}P]^2\,,
\end{equation}
for a set of orthonormal, traceless Hermitian operators $\PC$, it holds that
\begin{equation}\label{eq:group-witness}
    \mathbb{E}_{\HC}[\Lambda_{\rm qrt}(\ket{\psi_H})]=C\frac{\dim(\PC)}{2^n(2^n+1)}\,,
\end{equation}
which easily follows from the assumed properties of $\PC$ and the known trace equality $\Tr[(A\otimes B){\rm SWAP}]=\Tr[AB]$.
Notice that this is the case of all the QRTs considered in this paper except for the Clifford non-stabilizerness. Indeed, in all cases excluding Clifford non-stabilizerness, the proposed resourcefulness witness takes the form above for a set $\PC$ given by some subset of Pauli operators. Notice that the latter are orthogonal but non-normalized; i.e., for any two Pauli operators $P_i, P_j$ we have $\Tr[P_iP_j]=2^n\delta_{i,j}$, which immediately leads to the only formula we need for the witnesses studied in this manuscript
\begin{equation}\label{eq:haar-witness}
    \mathbb{E}_{\HC}[\Lambda_{\rm qrt}(\ket{\psi_H})]=C\frac{\dim(\PC)}{(2^n+1)}\,.
\end{equation}

\subsubsection{Entanglement}
The  set of operators $\PC_{\rm ent}$ defining the resource witness for the entanglement QRT in a Hilbert space of $n$ qubits has size $3n$. Using Eq.~\eqref{eq:haar-witness} it follows that 
\begin{equation}\label{eq:haar-ent}
    \mathbb{E}_{\HC}[\Lambda_{{\rm ent}}(\ket{\psi_H})] = \frac{3}{2^n+1}\,.
\end{equation}

\subsubsection{Fermionic non-Gaussianity}

In this case the set $\PC_{\rm ferm}$ is given by the $n$-qubit Pauli operators resulting from products of two distinct Majorana operators. Since there are $2n$ of the latter it follows that $\dim(\PC_{\rm ferm})=\binom{2n}{2}=n(2n-1)$, hence
\begin{equation}
    \mathbb{E}_{\HC}[\Lambda_{{\rm ferm}}(\ket{\psi_H})] = \frac{2n - 1}{2^n + 1}\,.
\end{equation}

\subsubsection{Imaginarity}
In this case, the set $\PC_{sym}$ is given by all the Pauli operators comprising an even number of $Y$'s. Its size can be readily found to be $\dim(\PC_{sym}) = \frac{1}{2}(4^n+2^n-2)$. Thus

\begin{equation}
    \mathbb{E}_{\HC}[\Lambda_{{\rm imag}}(\ket{\psi_H})] = \frac{2+2^n}{2+2^{n+1}}\,.
\end{equation}

\subsubsection{Realness}

For the QRT of realness, the set $\PC_{asym}$ is given by all the Pauli operators comprising an odd number of $Y$'s. One finds that $\dim(\PC_{asym}) = 2^{n-1}(2^n-1)$, leading to
\begin{equation}
    \mathbb{E}_{\HC}[\Lambda_{{\rm real}}(\ket{\psi_H})] = \frac{2^n-1}{2^n+1} = \tanh(\frac{n}{2}\ln(2))\,.
\end{equation}

\subsubsection{Spin coherence}

Notice that in the spin coherence QRT we have $\PC_{\rm coh} = \{S_x,S_y,S_z\}$, the three components of the $s$-spin operator, where $s=\frac{2^n-1}{2}$. Now these operators are not Pauli operators, thus we need to adjust the formula in Eq.~\eqref{eq:haar-witness}. Particularly, we now have $\Tr[S_iS_j]=\frac{1}{3}(4^n-1)2^{n-2}\delta_{i,j}$, hence together with $\dim(\PC_{\rm coh})=3$ we get that
\begin{equation}
    \mathbb{E}_{\HC}[\Lambda_{{\rm coh}}(\ket{\psi_H})] = \frac{(4^n-1)2^{n-2}}{2^n\left(\frac{2^n-1}{2}\right)^2(2^n+1)}=\frac{2^n-1}{(2^n+1)^2}\,.
\end{equation}

\subsubsection{Clifford non-stabilizerness}
The Haar average of the Clifford non-stabilizerness of a random $n$-qubit quantum state is more complicated, as $\Lambda_{\rm stab}(\ket{\psi})$ results from a summation of quartic overlaps $\Tr[\dya{\psi}P]^4$ over the set $\PC_{\rm stab}$ comprising all traceless Pauli operators. 
However, we can resort again to Weingarten calculus~\cite{mele2023introduction} and Eq.~\eqref{eq:weingarten} to find, on average over a Hilbert space $\HC$ of dimension $2^n$, the following quantity
\begin{equation}
    \mathbb{E}_{\HC}[\dya{\psi_H}^{\otimes t}]=\frac{P_{sym}^{(2^n, t)}}{\Tr[P_{sym}^{(2^n, t)}]}\,.
\end{equation}
Here $P_{sym}^{(2^n, t)} = \sum_{\pi \in S_t} R(\pi)$ where $S_t$ is the symmetric group of degree $t$ and where we recall that the representation $R(\pi)$  was defined above in Eq.~\eqref{eq:prem}. Thus, the average contributions to $\Lambda_{\rm stab}(\ket{\psi})$ will look like
\begin{equation}\label{eq-ap:haar_4th_moment}
    \mathbb{E}_{\HC}\Tr[\dya{\psi_H}P]^4 = \sum_{\pi \in S_4}\frac{\Tr[R(\pi) P^{\otimes 4}]}{\Tr[P_{sym}^{(2^n, 4)}]}\,.
\end{equation}
The elements of $S_4$ can be categorized by the type of cycle: 0-cycle (identity), 2-cycles, disjoint 2-cycles, 3-cycles and 4-cycles. The type of cycle determines $\Tr[R(\pi) P^{\otimes 4}]$.

\begin{align}
    \pi \text{ is a 0-cycle:} &\quad \Tr[R(\pi) P^{\otimes 4}] = 0 \,,\nonumber\\
    \pi \text{ is a 2-cycle:} &\quad \Tr[R(\pi) P^{\otimes 4}] = 0 \,,\nonumber\\
    \pi \text{ is a disjoint 2-cycle:} &\quad \Tr[R(\pi) P^{\otimes 4}] = 2^{2n}\,,\nonumber\\
    \pi \text{ is a 3-cycle:} &\quad \Tr[R(\pi) P^{\otimes 4}] = 0 \,,\nonumber\\
    \pi \text{ is a 4-cycle:} &\quad \Tr[R(\pi) P^{\otimes 4}] = 2^n\,.\nonumber
\end{align}

In $S_4$ there are three total disjoint 2-cycles and six total 4-cycles, and we can find that $\Tr[P_{sym}^{(d, t)}] = t! \binom{t+d-1}{t} = \frac{(t+d-1)!}{(d-1)!}$. For $t = 4$ and $d = 2^n$ we thus have $\Tr[P_{sym}^{(2^n, 4)}] = 2^n(2^n+1)(2^n+2)(2^n+3)$. We know that there are $4^n-1$ non-identity Pauli operators that constitute the set $\PC_{\rm stab}$, and each of these operators contributes $\frac{3 \cdot 2^{2n} + 6 \cdot 2^n}{2^{n}(2^n+1)(2^n+2)(2^n+3)}$. Adding up the terms and dividing by the normalization coefficient $2^n - 1$, leads to
\begin{equation}
    \mathbb{E}_{\HC}[\Lambda_{{\rm stab}}(\ket{\psi_H})] = \frac{3}{2^n+3}\,.
\end{equation}

\subsubsection{$S_n$-equivariance}

Recall that in the $S_n$-equivariance QRT, resourcefulness is measured by $\Lambda_{{S_n}}(\ket{\psi}) = \frac{1}{2^n-1}\sum_{P \in\PC_{{S_n}}} \Tr[\dya{\psi}P ]^{2}$, where $\PC_{S_n}$ consists of the normalized twirls of all the $4^n-1$ non-trivial Pauli operators over the $S_n$ group. Specifically, each Pauli $P$ gets mapped to $\frac{1}{n!}\sum_{\pi\in S_n}R(\pi)P R\ad(\pi)$, for $R$ the qubit permuting representation of $S_n$, and then gets normalized. 
Since conjugating by a permutation matrix maps a Pauli string into another one by permuting its components, it follows that the result of twirling a Pauli $P$ over $S_n$ only depends on the initial number of $X$'s, $Y$'s, and $Z$'s. By the stars and bars theorem~\cite{flajolet2009analytic}, the number of possible different assignments of non-trivial components for a Pauli string is $\binom{n+3}{3}-1$. 
This results in $\PC_{S_n}$ having size $\dim(\PC_{S_n})=Te_{n+1} - 1$ where $Te_m=\frac{1}{6}(m(m+1)(m+2))$ is the $m$-th tetrahedral number.

Now, given a representative Pauli operator with $n_x$ $X$ terms, $n_y$ $Y$ terms and $n_z$ $Z$ terms, the result of twirling over $S_n$ is the sum of the $\binom{n}{n_x}\binom{n-n_x}{n_y}\binom{n-n_x-n_y}{n_z}$ Pauli strings in the orbit of the representative element, each weighted by $\frac{n_x!n_y!n_z!(n-n_x-n_y-n_z)!}{n!}$. Orthogonality of the elements in $\PC_{{S_n}}$ follows easily from the fact that each orbit corresponds to a different assignment $(n_x,n_y,n_z)$. 
To normalize the twirls we just divide each element $P\in\PC_{S_n}$ by $\sqrt{\frac{\Tr[P^2]}{2^n}}=\sqrt{\frac{n_x!n_y!n_z!(n-n_x-n_y-n_z)!}{n!}}$, with $(n_x,n_y,n_z)$ the assignment of non-trivial Pauli operators in $P$.
Using Eq.~\eqref{eq:haar-witness}, we have that the average $S_n$-equivariance witness $\Lambda_{{S_n}}$ is given by
\begin{align}
    \mathbb{E}_{\HC}[\Lambda_{S_n}(\ket{\psi_H})] &= \frac{Te_{n+1}-1}{4^n-1}\,.
\end{align}

\subsubsection{Non-uniform entanglement}

In the QRT of non-uniform entanglement over a Hilbert space of $n$ qubits, the set $\PC_{\rm uent}$ consists of the three uniform weight-one operators $\PC_{\rm uent}=\{\sum_{i=1}^n X_i,\sum_{i=1}^n Y_i,\sum_{i=1}^n Z_i\}$. 
Notice that Eq.~\eqref{eq:haar-witness} does not directly apply here, since the elements of $\PC_{\rm uent}$ are not orthonormal. However they are readily found to be orthogonal, and one can check that $\Tr[P^2]=n2^n$ for each $P\in\PC_{\rm uent}$. Using this, one finds that Eq.~\eqref{eq:haar-witness} changes to
\begin{equation}
    \mathbb{E}_{\HC}[\Lambda_{\rm uent}(\ket{\psi_H})]=\frac{3}{n(2^n+1)}\,.
\end{equation}

\subsection{Expected witness values for random states in $\SC_{\rm imag}$, proof of results in Table~\ref{tab:Haar-tab}}

Here we compute the expected resource witnesses for random states in $\SC_{\rm imag}$, which are obtained by applying a random unitary from $\text{O}(d)$ to the reference state $\ket{0}^{\otimes n}$. As such, we henceforth define expectation values of random states in $\SC_{\rm imag}$ as $\mathbb{E}_{\SC_{\rm real}}[f(\dya{\psi})]=\mathbb{E}_{U\sim\text{O}(2^n)}[f(U\dya{0}^{\otimes n}U\ad)]$. 
We begin by recalling that, as discussed in Appendix~\ref{app:weingarten}, in order to compute expectation values over the Haar measure of a group, one needs to project into its commutant. For the special case of the orthogonal group, a basis of the $t$-th order commutant is given by a representation $F$~\cite{garcia2023deep} of the Brauer algebra $B_t$ acting on the $t$-fold tensor product Hilbert
space. That is,
\begin{equation}
    {\rm comm}^{(t)}(\text{O}(2^n))={\rm span}_{\mathbb C}\{F(\pi)\,,\forall \pi\in B_t\}.
\end{equation}
For convenience, we recall that the Brauer algebra is composed of all possible pairings on a set of $2t$ items.  Hence, the basis of
the commutant contains $\frac{(2t)!}{ 2^t (t!)}$ elements. 

From the previous we can find that for $t=2$ one has
\begin{equation}
    B_2=\{(\{1,3\},\{2,4\}),(\{1,4\},\{2,3\}),(\{1,2\},\{3,4\})\}\,,
\end{equation}
so that
\begin{equation}
    {\rm comm}^{(2)}(\text{O}(2^n))={\rm span}_{\mathbb C}\{\id\otimes \id,\SWAP,\Pi\}\,,
\end{equation}
where $\Pi=\sum_{i,j=1}^{2^n}\ket{ii}\bra{jj}$ is proportional to the projector onto the $2^n$-dimensional Bell state $\ket{\Phi}=\frac{1}{\sqrt{2^n}}\sum_{i=1}^{2^n}\ket{ii}$. Above, we have used the fact that 
\begin{equation}
    F((\{1,3\},\{2,4\}))=\id\otimes \id\,,\quad F((\{1,4\},\{2,3\}))=\SWAP\,,\quad F((\{1,2\},\{3,4\}))=\Pi\,.
\end{equation}
Combining the previous basis for the second order commutant of the orthogonal group with Eq.~\eqref{eq:weingarten} leads to (see also Appendix D in Ref.~\cite{garcia2023deep} for more details) 
\begin{align}\label{eq:general-O-av}
    \mathbb{E}_{U\sim\text{O}(2^n)}[U^{\otimes 2}X(U\ad)^{\otimes 2}]=&\frac{1}{2^n(2^n-1)(2^n+2)}\left((2^n+1)\Tr[X]-\Tr[X\SWAP]-\Tr[X\Pi]\right)\id\otimes \id\nonumber\\
    &+\frac{1}{2^n(2^n-1)(2^n+2)}\left(-\Tr[X]+(2^n+1)\Tr[X\SWAP]-\Tr[X\Pi]\right)\SWAP\nonumber\\
    &+\frac{1}{2^n(2^n-1)(2^n+2)}\left(-\Tr[X]-\Tr[X\SWAP]+(2^n+1)\Tr[X\Pi]\right)\Pi\,.
\end{align}
By replacing $X=\dya{0}^{\otimes 2n}$ we obtain that the average of the two-fold tensor product state in $\SC_{{\rm imag}}$ takes the form
\begin{equation}
  \mathbb{E}_{\SC_{\rm imag}}[\dya{\psi}^{\otimes 2}]=  \mathbb{E}_{U\sim \text{O}(2^n)]}[U^{\otimes 2}\dya{0}^{\otimes 2n}(U\ad)^{\otimes 2}]=\frac{2^n-1}{2^n(2^n-1)(2^n+2)}\left(\id\otimes \id+\SWAP+\Pi\right)\,.\label{eq:average-state-O}
\end{equation}

Equipped with Eq.~\eqref{eq:average-state-O} we can readily find that, for QRTs whose resourcefulness witnesses can be expressed as
\begin{equation}
    \Lambda_{\rm qrt}(\ket{\psi})=C\sum_{P\in\PC} \Tr[\dya{\psi}P]^2\,,
\end{equation}
for a set of orthonormal, traceless Hermitian operators $\PC$, it holds that
\begin{equation}\label{eq:group-witness-2}
    \mathbb{E}_{\SC_{\rm imag}}[\Lambda_{\rm qrt}(\ket{\psi})]=C\frac{2\dim(\PC^{({\rm sym})})}{2^n(2^n+2)}\,,
\end{equation}
where $\PC^{({\rm sym})}\subseteq \PC$ is the subset of symmetric operators within $\PC$. That is, $\PC^{({\rm sym})}=\{P\in \PC\,\,|\,\, P=P^T\}$. Here, we used the fact that $\Tr[P^{\otimes 2}\Pi]=\Tr[P P^T]$, and hence $\Tr[P^{\otimes 2}\Pi]+\Tr[P^{\otimes 2}\SWAP]$ is equal to $2$ if $P$ is symmetric, and $0$ if $P$ is antisymmetric. As before, if the elements in $\PC$ are not normalized to one, we need to multiply Eq.~\eqref{eq:group-witness-2} by their normalization.

\subsubsection{Entanglement}
The  set of operators $\PC_{\rm ent}$ defining the resource witness for the entanglement QRT has size $3n$, out of which $2n$ are symmetric. From Eq.~\eqref{eq:group-witness-2} we obtain  
\begin{equation}
    \mathbb{E}_{\HC}[\Lambda_{\rm ent}(\ket{\psi_H})]=\frac{4}{(2^n+2)}\,.
\end{equation}

\subsubsection{Fermionic non-Gaussianity}

In this case the set $\PC_{\rm ferm}$ is given by the $n$-qubit Pauli operators, so that $\dim(\PC_{\rm ferm})=\binom{2n}{2}=n(2n-1)$. Crucially, we recall that such a set of operators can be expressed as~\cite{diaz2023showcasing} $\PC_{\rm ferm}=\{Z_i\}_{i=1}^n \cup \{\widehat{X_iX_j},\widehat{X_iY_j},\widehat{Y_iX_j},\widehat{Y_iY_j}\}_{1\leq i<j\leq n}$, where $\widehat{A_iB_j}=A_i Z_{i+1}\cdots Z_{j-1}B_j$. As such, the $n$ operators $Z_i$ are symmetric, as well as the $n(n-1)$ operators of the form $\widehat{X_iX_j}$ and $\widehat{Y_iY_j}$. As such, we obtain
\begin{equation}
    \mathbb{E}_{\SC_{\rm imag}}[\Lambda_{\rm ferm}(\ket{\psi})]=\frac{2n}{(2^n+2)}\,.
\end{equation}

\subsubsection{Spin coherence}

For the QRT of spin coherence we have $\PC_{\rm coh} = \{S_x,S_y,S_z\}$, where $s=\frac{2^n-1}{2}$. Since these operators are not Pauli operators, we need to adjust the formula in Eq.~\eqref{eq:group-witness-2}. Particularly, we now have that $S_x$ and $S_z$ are symmetric, with $\Tr[S_iS_i]=\frac{1}{3}(4^n-1)2^{n-2}$ for all $i=x,z$. Hence, we find 

\begin{equation}
    \mathbb{E}_{\SC_{\rm imag}}[\Lambda_{{\rm coh}}(\ket{\psi_H})] = \frac{4(4^n-1)2^{n-2}}{3\left(\frac{2^n-1}{2}\right)^22^n(2^n+2)}=\frac{4}{3}\frac{2^n-1}{(2^n+1)^2}\,.
\end{equation}

\subsubsection{Clifford non-stabilizerness}
 
When computing the expectation value $\mathbb{E}_{\SC_{\rm  imag}} \left[ \Lambda_{{\rm stab}}(\ket{\psi})\right]$, we need to compute the Gram and the Weingarten matrix for the Brauer algebra $B_t$, which is a matrix of size $105\times 105$. While such analysis can be cumbersome, it is however still tractable via standard software (e.g. Mathematica). Then, when using Eq.~\eqref{eq:weingarten} one needs to evaluate the quantities $\Tr[\dya{0}^{\otimes 4n} F(\sigma)]$ and $\Tr[P^{\otimes 4} F(\sigma)]$ for all $\sigma\in B_4$. A straightforward calculation reveals that 
\begin{equation}
    \Tr[\dya{0}^{\otimes 4n} F(\sigma)]=1\,,\quad \forall \sigma\in B_4\,.
\end{equation}
Then,  for any antisymmetric Pauli  $P$ 
\begin{equation}\label{eq:brauer-tp4}
    \Tr[P^{\otimes 4} F(\sigma)]=\begin{cases}
        0\,,\quad\text{if $\{i,i+4\}\in \sigma$ for any $i\in1,\ldots,4$ }\,,\\
        -d^2\,,\quad\text{if $\sigma$ has two cycles, one being a transposition} \,,\\
        d^2\,,\quad\text{if $\sigma$ has two cycles none being a transposition} \,,\\
        -d\,,\quad\text{if $(\{i,j\},\{k,i+n\})\in \sigma$ with $i,k\leq n$ and $j>n$ and $\sigma\notin S_4$} \,,\\
        d\,,\quad\text{else}\,,
    \end{cases}
\end{equation}
whereas for any symmetric $P$
\begin{equation}\label{eq:brauer-tp42}
    \Tr[P^{\otimes 4} F(\sigma)]=\begin{cases}
        0\,,\quad\text{if $\{i,i+4\}\in \sigma$ for any $i\in1,\ldots,4$ }\,,\\
        d^2\,,\quad\text{if $\{i,i+4\}\notin \sigma$ for any $i\in1,\ldots,4$ and $\sigma$ has two cycles } \,,\\
        d\,,\quad\text{else}\,.
    \end{cases}
\end{equation}
Combining the previous results, we find
\begin{equation}
    \mathbb{E}_{\SC_{\rm  imag}} \left[ \Lambda_{{\rm stab}}(\ket{\psi})\right]=\frac{6}{6 + 2^n}\,.
\end{equation}

\subsubsection{$S_n$-equivariance}

Here, we begin by counting how many $S_n$-equivariant Paulis  there are with an even number of $Y$'s. In particular, we can obtain this result from the summation
\small
\begin{equation}
    \left(\sum_{n_x=0}^n\sum_{n_z=0}^{n-n_x}\sum_{n_y=0,2,\ldots,n-n_x}1\right)-1=\frac{1}{6}\left(3n(3+n)+3\left\lfloor\frac{(n-1)}{2}\right\rfloor^2+2\left\lfloor\frac{(n-1)}{2}\right\rfloor^3+4\left\lfloor\frac{n}{2}\right\rfloor+2\left\lfloor\frac{n}{2}\right\rfloor^3+\left\lfloor\frac{(n-1)}{2}\right\rfloor\left(1+6\left\lfloor\frac{n}{2}\right\rfloor\right)\right).\nonumber
\end{equation}
\normalsize
Assuming that $n$ is even, the previous simplifies to 
\begin{equation}\label{eq:numb-sn-sym}
    \left(\sum_{n_x=0}^n\sum_{n_z=0}^{n-n_x}\sum_{n_y=0,2,\ldots}^{n-n_x-nz}1\right)-1=\frac{1}{6}\left(\frac{n}{2}+1\right)\left(\frac{n}{2}+2\right)(2n+3)-1\,.
\end{equation}
Combining this result with the Eq.~\eqref{eq:group-witness-2} leads to
\begin{equation}
    \mathbb{E}_{\SC_{\rm imag}}[\Lambda_{ S_n}(\ket{\psi})]=\frac{2\left(\frac{1}{6}(\frac{n}{2}+1)(\frac{n}{2}+2)(2n+3)-1\right)}{(2^n-1)(2^n+2)}\,.
\end{equation}

\subsubsection{Non-uniform entanglement}
For the QRT of non-uniform entanglement, the set $\PC_{\rm uent}$ consists of the three uniform weight-one operators $\PC_{\rm uent}=\{\sum_{i=1}^n X_i,\sum_{i=1}^n Y_i,\sum_{i=1}^n Z_i\}$, with the first and the last being symmetric. Using the fact that  $\Tr[P^2]=n2^n$ for each $P\in\PC_{\rm uent}$, we obtain from Eq.~\eqref{eq:group-witness-2}
\begin{equation}
    \mathbb{E}_{\SC_{\rm imag}}[\Lambda_{\rm uent}(\ket{\psi})]=\frac{4}{n(2^n+2)}\,.
\end{equation}

\subsection{Expected witness values for random states in $\SC_{\rm real}$, proof of results in Table~\ref{tab:Haar-tab}}

Here we study the expected resource witness values for the free states in the QRT of realness. We recall that the states in $\SC_{\rm real}$ can be obtained by applying a unitary from $\text{O}(2^n)$ to the reference state $\ket{+_y}^{\otimes n}$. As such, we henceforth define expectation values of random states in $\SC_{\rm real}$ as $\mathbb{E}_{\SC_{\rm real}}[f(\dya{\psi})]=\mathbb{E}_{U\sim\text{O}(2^n)}[f(U\dya{+_y}^{\otimes n}U\ad)]$. Then, $\mathbb{E}_{U\sim\text{O}(2^n)}$ denotes the average over the orthogonal group's Haar measure. From the previous, we can use the results in Eq.~\eqref{eq:general-O-av} to obtain
\begin{align}\label{eq:average-plusstate-O}
    \mathbb{E}_{U\sim\text{O}(2^n)}[U^{\otimes 2}\dya{+_y}^{\otimes 2n}(U\ad)^{\otimes 2}]=&\frac{1}{(2^n-1)(2^n+2)}\left(\id\otimes\id +\SWAP-\frac{1}{2^{n-1}}\Pi\right)\,,
\end{align}
where we used the fact that $\Tr[\dya{+_y}^{\otimes 2n}\Pi]=\Tr[(\dya{+_y}^{\otimes n})(\dya{+_y}^{\otimes n})^T]=0$.

From Eq.~\eqref{eq:average-plusstate-O} we obtain that, for QRTs whose resourcefulness witnesses can be expressed as
\begin{equation}
    \Lambda_{\rm qrt}(\ket{\psi})=C\sum_{P\in\PC} \Tr[\dya{\psi}P]^2\,,
\end{equation}
for a set of orthonormal, traceless Hermitian operators $\PC$, then
\begin{equation}\label{eq:group-witness-3}
    \mathbb{E}_{\SC_{\rm real}}[\Lambda_{\rm qrt}(\ket{\psi})]=C\frac{1}{(2^n-1)(2^n+2)}\left(\left(1-\frac{1}{2^{n-1}}\right)\dim(\PC^{({\rm sym})})+\left(1+\frac{1}{2^{n-1}}\right)\dim(\PC^{({\rm asym})})\right)\,,
\end{equation}
where $\PC^{({\rm asym})}\subseteq \PC$ is the subset of antisymmetric operators within $\PC$. That is, $\PC^{({\rm asym})}=\{P\in \PC\,\,|\,\, P=-P^T\}$.

\subsubsection{Entanglement}
We recall that the set of operators $\PC_{\rm ent}$ has size $3n$, out of which $2n$ are symmetric and $n$ are antisymmetric. From Eq.~\eqref{eq:group-witness-3} we obtain  
\begin{equation}
    \mathbb{E}_{\SC_{\rm real}}[\Lambda_{\rm ent}(\ket{\psi_H})]=\frac{2^n}{n(2^n-1)(2^n+2)}\left(\left(1-\frac{1}{2^{n-1}}\right)2n+\left(1+\frac{1}{2^{n-1}}\right)n\right)=\frac{3\cdot 2^n-2}{(2^n-1)(2^n+2)}\,.
\end{equation}

\subsubsection{Fermionic non-Gaussianity}

Next, consider $\PC_{\rm ferm}$ which contains $n^2$ symmetric operators and $n(n-1)$ antisymmetric ones. Using Eq.~\eqref{eq:group-witness-3} we reach
\begin{equation}
    \mathbb{E}_{\SC_{\rm real}}[\Lambda_{\rm ferm}(\ket{\psi})]=\frac{2^n}{n(2^n-1)(2^n+2)}\left(\left(1-\frac{1}{2^{n-1}}\right)n^2+\left(1+\frac{1}{2^{n-1}}\right)n(n-1)\right)=\frac{2^n(2n-1)-2}{(2^n-1)(2^n+2)}\,.
\end{equation}

\subsubsection{Spin coherence}

Now, when considering $\PC_{\rm coh} = \{S_x,S_y,S_z\}$, two of these operators are symmetric, whereas one is antisymmetric. Adapting the normalization correctly,  

\begin{equation}
\mathbb{E}_{\SC_{\rm real}}[\Lambda_{\rm coh}(\ket{\psi})]=\frac{(4^n-1)2^{n-2}}{3\left(\frac{2^n-1}{2}\right)^2(2^n-1)(2^n+2)}\left(2\left(1-\frac{1}{2^{n-1}}\right)+\left(1+\frac{1}{2^{n-1}}\right)\right)=\frac{3\cdot 2^{2n}+2^n-2}{3(2^n-1)^2(2^n+2)}\,.
\end{equation}

\subsubsection{Clifford non-stabilizerness}
As in the case for the in $\SC_{\rm imag}$, we here construct the full Weingarten matrix. By combining Eq.~\eqref{eq:brauer-tp4} along with 
\begin{equation}
    \Tr[\dya{+_y}^{\otimes 4n} F(\sigma)]=\begin{cases}
        1\,,\quad \forall \sigma\in S_4\subseteq B_4\,,\\
        0\,,\quad \text{otherwise}\,,
    \end{cases}
\end{equation}
we obtain
\begin{equation}
\mathbb{E}_{\SC_{\rm real}}[\Lambda_{\rm stab}(\ket{\psi})]=\frac{3 \left(3\cdot 2^n+4^n-2\right)}{\left(2^n-1\right) \left(2^n+1\right) \left(2^n+6\right)}\,.
\end{equation}

\subsubsection{$S_n$-equivariance}

Let us assume for simplicity $n$ even. Combining Eqs.~\eqref{eq:numb-sn-sym} and~\eqref{eq:group-witness-3} leads to  
\begin{equation}
    \mathbb{E}_{\SC_{\rm real}}[\Lambda_{S_n}(\ket{\psi})]=\frac{2^nn(n(n+6)+11)-3(n(n+4)+8)}{6(2^n-1)^2(2^n+2)}\,.
\end{equation}

\subsubsection{Non-uniform entanglement}
A straightforward calculation using Eq.~\eqref{eq:group-witness-3} leads to
\begin{equation}
    \mathbb{E}_{\SC_{\rm real}}[\Lambda_{\rm uent}(\ket{\psi})]=\frac{3\cdot 2^n-2}{n(2^n-1)(2^n+2)}\,.
\end{equation}

\subsection{Expected witness values for random tensor product state in $\SC_{\rm ent}$, proof of Proposition~\ref{prop:haar-sep}}

In this section we provide proofs for the results presented in Proposition~\ref{prop:haar-sep} of the main text, regarding the average resourcefulness witnesses of various QRTs over $n$-qubit product states $\ket{\psi}=\bigotimes_{j=1}^n\ket{\psi_j}$, where each state $\ket{\psi_j}$ is sampled independently according to the Haar measure  over $\HC_j=\mathbb{C}^2$. Given that we can re-write such random state as $\ket{\psi}=\bigotimes_{j=1}^nU_j\ket{0}$, where  now each unitary $U_j$ is sampled independently according to the Haar measure  over $\text{U}(2)$, which enables the computation of expectation values via the Weingarten calculus and Eq.~\eqref{eq:weingarten}. That is, $\mathbb{E}_{\SC_{\rm ent}}[f(\dya{\psi})]=\mathbb{E}_{\HC_1}\cdots\mathbb{E}_{\HC_n}[f(\prod_j\dya{\psi_j})]=\mathbb{E}_{U_1\sim \text{U}(2)}\cdots\mathbb{E}_{U_n\sim \text{U}(2)}[f(\bigotimes_jU_j\dya{0}U_j\ad)]$. 

Indeed, by our assumption of independence in the sampling of the local unitaries, we can still make use of Eq.~\eqref{eq:Haar} and obtain
\begin{equation}\label{eq:local-Haar}
    \mathbb{E}_{\SC_{\rm ent}}[\dya{\psi}^{\otimes 2}]=\bigotimes_{j=1}^n\frac{\id_j\otimes \id_j +{\rm SWAP_j}}{6}\,,
\end{equation}
where now ${\rm SWAP}_j$ swaps the two copies of the $j$-th Hilbert space $\HC_j$.
We will employ this result as a building block of our proof for most of the following results.

\subsubsection{Entanglement}

Trivially, any state of the form $\ket{\psi}=\bigotimes_{j=1}^n\ket{\psi_j}$ is a free state for the QRT of entanglement, i.e. $\ket{\psi}\in\SC_{\rm ent}$. Thus 
\begin{equation}
    \mathbb{E}_{\SC_{\rm ent}}[\Lambda_{\rm ent}(\ket{\psi})]=\mathbb{E}_{\SC_{\rm ent}}[1]=1\,.
\end{equation}

\subsubsection{Fermionic non-Gaussianity}

We need to compute
\begin{equation}
    \mathbb{E}_{\SC_{\rm ent}}[\Lambda_{\rm ferm}(\ket{\psi})]
    =\mathbb{E}_{\SC_{\rm ent}}\!\Bigl[\frac{1}{n}\sum_{P\in\PC_{\rm ferm}}\Tr\!\bigl[\dya{\psi}P\bigr]^2\Bigr]\,.
\end{equation}
Using standard properties of the trace and Eq.~\eqref{eq:local-Haar}, this becomes
\begin{equation}
    \mathbb{E}_{\SC_{\rm ent}}[\Lambda_{\rm ferm}(\ket{\psi})]
    =\frac{1}{n}\sum_{P\in\PC_{\rm ferm}}\!\prod_{j=1}^n\frac{1}{6}
      \Tr\!\bigl[(\id_j^{\otimes 2}+{\rm SWAP}_j)\,P_j^{\otimes 2}\bigr]\,,
\end{equation}
where $P_j$ denotes the $j$-th tensor factor of the Pauli string $P$.  

We now evaluate each factor:
\begin{equation}
    \Tr\!\bigl[{\rm SWAP}_j\,P_j^{\otimes 2}\bigr]=\Tr[P_j^2]=2,
    \qquad
    \Tr\!\bigl[\id_j^{\otimes 2}\,P_j^{\otimes 2}\bigr]=\Tr[P_j^{\otimes 2}]=4\,\delta_{P_j,\id}\,.
\end{equation}

Recall that the Pauli strings $P$ belonging to $\PC_{\rm ferm}$ consist of single-site strings $P=Z_a$ for $a\in[n]$, and two-site strings of the form
\begin{equation}
    X_aZ_{a+1}\cdots Z_{b-1}X_b,\quad
    X_aZ_{a+1}\cdots Z_{b-1}Y_b,\quad
    Y_aZ_{a+1}\cdots Z_{b-1}X_b,\quad
    Y_aZ_{a+1}\cdots Z_{b-1}Y_b\,,
    \quad 1\le a<b\le n\,.
\end{equation}
Each single-site string $Z_a$ contributes
\begin{equation}
    \prod_{j=1}^n\frac{1}{6}\Tr\!\bigl[(\id_j^{\otimes 2}+{\rm SWAP}_j)\,P_j^{\otimes 2}\bigr]
    =\frac{1}{3}\,,
\end{equation}
while the other four possibilities contribute $
    \Bigl(\frac{1}{3}\Bigr)^{b-a+1}$. Adding all contributions yields
\begin{equation}
    \mathbb{E}_{\SC_{\rm ent}}[\Lambda_{\rm ferm}(\ket{\psi})]
    =\frac{1}{n}\!\Bigl[\frac{n}{3}
      +4\sum_{a=1}^{n-1}\sum_{b=a+1}^n\Bigl(\tfrac{1}{3}\Bigr)^{b-a+1}\Bigr]
    =\frac{n-1+3^{-n}}{n}\,,
\end{equation}
recovering the result presented in the main text.

\subsubsection{Imaginarity}

We can again make use of Eq.~\eqref{eq:local-Haar} to find
\begin{equation}
    \mathbb{E}_{\SC_{\rm ent}}[\Lambda_{\rm imag}(\ket{\psi})]
    =\frac{1}{2^n-1}\sum_{P\in\PC_{\rm imag}}\!\prod_{j=1}^n\frac{1}{6}
      \Tr\!\bigl[(\id_j^{\otimes 2}+{\rm SWAP}_j)\,P_j^{\otimes 2}\bigr]\,.
\end{equation}

Again, one has that any identity in $P$ contributes one to the product, while non-trivial Pauli terms contribute $1/3$. Hence, we are left with ordering all the Pauli strings appearing in $\PC_{\rm imag}$ by their bodyness. Since the latter consist of all the strings with an even number of $Y$'s, denoting $n_i$ the number of trivial components in $P$, $n_x$ that of $X$'s and $n_y$ that of $Y$'s, we have
\begin{equation}
    \mathbb{E}_{\SC_{\rm ent}}[\Lambda_{\rm imag}(\ket{\psi})] =
    \frac{1}{2^n-1}\sum_{n_i=0}^{n-1}\sum_{n_y = 0}^{\lfloor \frac{n-n_i}{2}\rfloor}\sum_{n_x=0}^{n-n_i-n_y}
    \binom{n}{n_i}\binom{n-n_i}{2n_y}\binom{n-n_i-n_y}{n_x}\left(\frac{1}{3}\right)^{n-n_i} = 
    \frac{6^n+4^n-2\cdot3^n}{2\cdot3^n(2^n-1)}\,,
\end{equation}
which is the result declared in the main text.

\subsubsection{Realness}

This case is completely analogous to that of the QRT of imaginarity. Indeed, we can simply change the normalization factor and consider the Pauli strings with an odd number of $Y$'s (i.e., $\PC_{\rm real}$) to get
\begin{equation}
    \mathbb{E}_{\SC_{\rm ent}}[\Lambda_{\rm real}(\ket{\psi})] =
    \frac{1}{2^{n-1}}\sum_{n_i=0}^{n-1}\sum_{n_y = 1}^{\lfloor \frac{n-n_i}{2}\rfloor}\sum_{n_x=0}^{n-n_i-n_y}
    \binom{n}{n_i}\binom{n-n_i}{2n_y-1}\binom{n-n_i-n_y}{n_x}\left(\frac{1}{3}\right)^{n-n_i} = 
    1-\left(\frac{2}{3}\right)^n\,,
\end{equation}
as we reported in the main text.
The careful reader can check that the relation $\Lambda_{{\rm imag}}(\rho) - 
\frac{\Lambda_{{\rm real}}(\rho)}{2^{1-n} - 2} = 1$ is indeed satisfied by the results presented here.

\subsubsection{Clifford non-stabilizerness}

The case of the Clifford non-stabilizerness QRT is again trickier due to the quartic dependence from the overlaps $\Tr[\dya{\psi}P]^4$ of $\Lambda_{\rm stab}(\ket{\psi})$. 
Let us manipulate the expression for $\Lambda_{\rm stab}(\ket{\psi})$ as follows
\begin{align}
    \Lambda_{\rm stab}(\ket{\psi}) &= \frac{1}{2^n -1}\sum_{P\in \PC_{\rm stab}}\Tr[\dya{\psi}P]^4 =\frac{1}{2^n -1}\sum_{P\in \PC_{\rm stab}}\Tr[\dya{\psi}^{\otimes 4}P^{\otimes 4}] =\frac{1}{2^n -1}\sum_{P\in \PC_{\rm stab}}\prod_{j=1}^n\Tr[\dya{\psi_j}^{\otimes 4}P_j^{\otimes 4}]\,,\nonumber
\end{align}
where we used standard properties of the trace, and plugged in the considered product state $\ket{\psi}=\bigotimes_{j=1}^n\ket{\psi_j}$.
Now, we can use again the result from Weingarten calculus $
    \mathbb{E}_{\HC}[\dya{\psi}^{\otimes t}]=\frac{P_{sym}^{(d, t)}}{\Tr[P_{sym}^{(d, t)}]}$, by setting $d=2$, and $t=4$. We find $\Tr[P_{sym}^{(2, 4)}]=120$, leading to
\begin{equation}
    \mathbb{E}_{\SC_{\rm ent}}[\Lambda_{\rm stab}(\ket{\psi})] = 
    \frac{1}{2^n -1}\sum_{P\in \PC_{\rm stab}}\prod_{j=1}^n\frac{\Tr[P_{sym}^{(2, 4)}P_j^{\otimes 4}]}{120}\,.
\end{equation}
From the analysis carried out below Eq.~\eqref{eq-ap:haar_4th_moment} we know that each non-trivial Pauli component $P_j$ will yield a contribution $1/5$. If $P_j=\id$ instead, the local expectation value corresponds to the expectation value of the fourth power of the trace of the pure density matrix $\dya{\psi_j}$, and hence is equal to one.
We are thus left with grouping the $4^n-1$ non-trivial Pauli operators in $\PC_{\rm stab}$ by their weight $k$ and multiplying their associated contribution $\left(\frac{1}{5}\right)^k$, which results in
\begin{equation}
    \mathbb{E}_{\SC_{\rm ent}}[\Lambda_{\rm stab}(\ket{\psi})] = 
    \frac{1}{2^n -1}\sum_{k=1}^n \binom{n}{k}3^k\left(\frac{1}{5}\right)^k=\frac{\left(\frac{8}{5}\right)^n-1}{2^n-1}\,,
\end{equation}
proving the result provided in the main text.

\subsubsection{$S_n$-equivariance}

We can once again use Eq.~\eqref{eq:local-Haar} to express the average resourcefulness with respect to the QRT of $S_n$-equivariance of an $n$ qubit product state $\ket{\psi}=\bigotimes_{j=1}^n\ket{\psi_j}$ as
\begin{equation}
    \mathbb{E}_{\SC_{\rm ent}}[\Lambda_{S_n}(\ket{\psi})]
    =\frac{1}{2^n-1}\sum_{P\in\PC_{S_n}}\!
      \Tr\!\left[\bigotimes_{j=1}^n\left(\frac{\id_j^{\otimes 2}+{\rm SWAP}_j}{6}\right)\,P^{\otimes 2}\right]\,.
\end{equation}
Let us recall from the previous section that each element in $\PC_{S_n}$ corresponds to the, normalized orbit of a representative Pauli string with a given assignment $(n_x,n_y,n_z)$ of $X,Y,Z$ components. Let us call $q=n_x+n_y+nz$ the bodyness of the Pauli strings in a given orbit.
We now study the terms $\Tr[{\rm SWAP}^{\otimes k}\id^{\otimes 2(n-k)}P^{\otimes 2}]$, where $P$ is an element of $\PC_{S_n}$ corresponding to the assignment $(n_x,n_y,n_z)$. When $q>k$ this term is bound to vanish, as there are not enough ${\rm SWAP}$ operators to compensate for the traceless Pauli components. On the other hand, for $q\leq k$, we get a non-vanishing contribution. Particularly, only the Pauli strings that are a tensor-square will contribute, since the same Pauli component is needed on the copies of the qubits acted upon by each ${\rm SWAP}$. Then, carrying out each trace shows that each bare Pauli string in $P$ whose $q$ non-trivial components are in the first $k$ slots contributes $2^k\cdot 4^{n-k}$. Considering their coefficient and normalization, and counting the number of valid strings in $P$ one finds
\begin{equation}
    \Tr[{\rm SWAP}^{\otimes k}\id^{\otimes 2(n-k)}P^{\otimes 2}] = \begin{cases}
        0 &\text{if $q>k$}\\
        2^k\cdot 4^{n-k}\cdot \frac{k!\,(n-q)!}{n!\,(k-q)!}&\text{if $q\leq k$}
    \end{cases}\,.
\end{equation}
Notice that, coherently with the $S_n$-equivariance of the operators involved, the contribution of each orbit only depends on $q$.
Furthermore, the $S_n$-equivariance of $P$ also implies that the previous contribution does not depend on which $k$ pairs of qubits are targeted by the ${\rm SWAP}$ operators. Thus, we can finally write
\begin{equation}
    \mathbb{E}_{\SC_{\rm ent}}[\Lambda_{S_n}(\ket{\psi})]
    =\frac{1}{6^n(2^n-1)}\sum_{k=1}^n\sum_{q=1}^k S(q,3)\cdot\binom{n}{k}\cdot2^k\cdot 4^{n-k}\cdot \frac{k!\,(n-q)!}{n!\,(k-q)!}\,,
\end{equation}
where $S(q,3)=\binom{q+2}{2}$ is the Stirling number of second kind, counting how many assignments $(n_x,n_y,n_z)$ with fixed $q$ there are, while the binomial factor $\binom{n}{k}$ arises from counting the number of $k$ ${\rm SWAP}$ operators. 
Carrying out the summation one finds the simplified expression reported in the main text
\begin{equation}
    \mathbb{E}_{\SC_{\rm ent}}[\Lambda_{S_n}(\ket{\psi})]
    =\frac{19 \cdot(3^{n}-1)-2n(n+6)}{8\cdot 3^n (2^n-1)}\,.
\end{equation}

\subsubsection{Non-uniform entanglement}

Resorting again to Eq.~\eqref{eq:local-Haar} we have
\begin{equation}
    \mathbb{E}_{\SC_{\rm ent}}[\Lambda_{\rm uent}(\ket{\psi})]
    =\frac{1}{n^2}\sum_{P\in\PC_{\rm uent}}\!
      \Tr\!\left[\bigotimes_{j=1}^n\left(\frac{\id_j^{\otimes 2}+{\rm SWAP}_j}{6}\right)\,P^{\otimes 2}\right]\,.
\end{equation}
Here $\PC_{\rm uent}=\{\sum_{i=1}^n X_i,\sum_{i=1}^n Y_i,\sum_{i=1}^n Z_i\}$. One can readily check that the only non-vanishing contributions to the average non-uniform entanglement resourcefulness come from the tensor-square terms, such as $X_1^{\otimes2}$. This is because the local ${\rm SWAP}_j$ operators annihilate any operator that is not a tensor-square. 
Hence, we can replace the $(\sum_{i=1}^nX_i)^{\otimes2}$ with $\sum_{i=1}^n X_i^{\otimes 2}$ and analogously for the other two elements of $\PC_{\rm uent}$.

Recalling from the previous derivations that $\Tr\!\left[\frac{\id_j^{\otimes 2}+{\rm SWAP}_j}{6}\,P_j^{\otimes 2}\right]=1$ if $P_j=\id$, and $\Tr\!\left[\frac{\id_j^{\otimes 2}+{\rm SWAP}_j}{6}\,P_j^{\otimes 2}\right]=\frac{1}{3}$ otherwise, we get to

\begin{equation}
    \mathbb{E}_{\SC_{\rm ent}}[\Lambda_{\rm uent}(\ket{\psi})]
    =\frac{1}{n^2} 3 \sum_{i=1}^n \frac{1}{3} = \frac{1}{n}\,,
\end{equation}
recovering the result stated in the main text.

\subsection{Fermionic entanglement for tensor product states, proof of Propositions~\ref{prop:uent-ferm-ent} and~\ref{prop:ent-fer-1}}

Here we provide a proof for Proposition~\ref{prop:uent-ferm-ent}. Without loss of generality, we can parametrize a general $n$-qubit uniform tensor product state as $\ket{\psi}=(R_z(\a)R_y(\b)\ket{0})^{\otimes n}$, where we use the fact that  any single qubit unitary $U\in\mathbb{SU}(2)$ can be decomposed in terms of three Euler angles as $U = R_z(\a)R_y(\b)R_z(\eta)$, with the rotation $R_z(\eta)$ leading to an unimportant global phase. Furthermore, acting with this unitary on the all zero reference state allows to reach any point in the Bloch sphere $\mathbb{C}^2$. 

Then, let us now notice that $R_z(\a)$ is a free operator within the fermionic non-Gaussianity QRT~\cite{diaz2023showcasing} (e.g., a rotation about the $z$-axis on the $i$-the qubit can be expressed as $e^{i \eta \gamma_{2i-1}\gamma_{2i}}$), so we know that the action of $R_z(\a)$ gates do not change the value of $\Lambda_{\rm ferm}$. We can hence focus on the state $\ket{\psi} = (R_y(\b)\ket{0})^{\otimes n}$. The density matrix of a single qubit uniform state can be thus expanded as 
\begin{equation}
    R_y(\b)\dya{0}R_y^\dagger(\b) = \left(\cos\left(\frac{\b}{2}\right)\id - i \sin\left(\frac{\b}{2}\right)Y\right)\left(\frac{\id+Z}{2}\right)\left(\cos\left(\frac{\b}{2}\right)\id + i \sin\left(\frac{\b}{2}\right)Y\right) = \frac{1}{2}(\id + \cos(\b)Z + \sin(\b)X) \nonumber \,.
\end{equation}
This leads to the following expression for the $n$-qubit uniform tensor product state
\begin{equation}
   \dya{\psi} = \frac{1}{2^n} \bigotimes_{j = 1}^n (\id_j + \cos(\b)Z_j + \sin(\b)X_j)\,,
\end{equation}
which consists of Pauli operators from $\{\id, X, Z\}^{\otimes n}$.
Recall that $\Lambda_{{\rm ferm}}(\ket{\psi})=\frac{1}{n}\sum_{P\in\PC_{{\rm ferm}} }\Tr[\dya{\psi}P]^2$, where $\PC_{\rm ferm}=\{i\gamma_j\gamma_k\}_{1\leq j<k\leq 2n}$ is the set of Pauli operators given by the product of two distinct Majorana operators.
One can check that the only Pauli operators in $\PC_{\rm ferm}\cap \{\id, X, Z\}^{\otimes n}$ are those of the form $P=i\gamma_{2a}\gamma_{2b-1}$ for $1\leq a<b\leq n$, which have the form $P\propto X_aZ_{a+1}\dots Z_{b-1}X_b$, and those reading $P=i\gamma_{2c-1}\gamma_{2c}=Z_c$ for $c=1,\dots,n$.
Explicitly computing the squares of the associated traces in $\Lambda_{\rm ferm}$ reveals that the latter Pauli operators contribute $\cos^2(\b)$, while the former contribute $\sin^4(\b)\cos^{2(b-a-1)}(\b)$.
Counting and adding up all the contributions leads to the formula presented in the main text
\begin{equation}
    \Lambda_{{\rm ferm}}(\ket{\psi}) = \frac{n + \cos^{2n}(\b) - 1}{n}\,.
\end{equation}

Notice that the minimum value of this equation occurs at $\b = \frac{\pi}{8}+k\frac{\pi}{4}$ for any integer $k$. Particularly, the minimum fermionic non-Gaussianity of uniform tensor product states reads
\begin{equation}
    \min_{\ket{\psi}\in\SC_{{\rm uent}}}\{\Lambda_{{\rm ferm}}(\ket{\psi})\} = \frac{n - 1}{n}.
\end{equation}

We now show that this minimum actually holds for tensor product but non-uniform states as well. 
By the same arguments used before, to study the latter we can consider $\ket{\psi} = \bigotimes_{j=1}^n R_y(\b_j)\ket{0}$. Then the same steps now lead to
\begin{equation}
    \Lambda_{{\rm ferm}}(\ket{\psi}) = \frac{1}{n}\left(\sum_{j=1}^n\cos^2(\b_j) + \sum_{l=1}^{n-1}\sum_{r=l+1}^{n}\sin^2(\b_l)\sin^2(\b_r)\prod_{m=l+1}^{r-1}\cos^2(\b_m)\right)\,.
\end{equation}
We can prove that the minimum is again $\min_{\ket{\psi}\in\SC_{{\rm ent}}}\{\Lambda_{{\rm ferm}}(\ket{\psi})\} = \frac{n - 1}{n}$ by  induction. We define
\begin{equation}
    f_n = \sum_{j=1}^n\cos^2(\b_j) + \sum_{l=1}^{n-1}\sum_{r=l+1}^{n}\sin^2(\b_l)\sin^2(\b_r)\prod_{m=l+1}^{r-1}\cos^2(\b_m)\,,
\end{equation}
and we make the claim 
\begin{equation}
    f_n = n-1+\prod_{j=1}^n\cos^2(\b_j)\,.
\end{equation}
The case $n=1$ is trivial, as $f_1=\cos^2(2\b_1)$.
Let us hence assume $f_{n-1}=n-2+\prod_{j=1}^{n-1}\cos^2(\b_j)$. We can write
\begin{equation}
    f_n=f_{n-1}+\cos^2(\b_n) + \sin^2(\b_n)\sum_{l=1}^{n-1}\sin^2(\b_l)\prod_{m=l+1}^{r-1}\cos^2(\b_m)\,.
\end{equation}
We now notice that, defining $P_k = \prod_{m=k}^{n-1}\cos^2(\b_m)$, with $P_n=1$, one finds
\begin{align}
    P_{l}-P_{l+1}= \cos^2(\b_l)\prod_{m=l+1}^{n-1}\cos^2(\b_m) - \prod_{m=l+1}^{n-1}\cos^2(\b_m)=-\sin^2(\b_l)\prod_{m=l+1}^{n-1}\cos^2(\b_m)\,.
\end{align}
Hence
\begin{equation}
    \sum_{l=1}^{n-1}\sin^2(\b_l)\prod_{m=l+1}^{n-1}\cos^2(\b_m)= -\sum_{l=1}^{n-1}(P_{l}-P_{l+1}) = P_n-P_1 = 1-\prod_{m=1}^{n-1}\cos^2(\b_m)\,.
\end{equation}
Thus, we have
\begin{align}
    f_n & = f_{n-1}+\cos^2(\b_n) + \sin^2(\b_n)\left(1-\prod_{m=1}^{n-1}\cos^2(\b_m)\right)\nonumber\\
    &= f_{n-1}+1-\sin^2(\b_n)\prod_{m=1}^{n-1}\cos^2(\b_m)\nonumber\\
    &= f_{n-1}+1-\prod_{m=1}^{n-1}\cos^2(\b_m) +\prod_{m=1}^{n}\cos^2(\b_m)\nonumber\\
    &= n - 1 + \prod_{m=1}^{n}\cos^2(\b_m)\,,
\end{align}
where in the last line we used the inductive assumption, proving our claim.
From this expression for $f_n$ it is immediate to see that the minimum fermionic non-Gaussianity resourcefulness for product states is indeed $\frac{n - 1}{n}$, which is attained when at least one angle $\b_j=\frac{\pi}{8}+k\frac{\pi}{4}$ for any integer $k$.

Lastly, let us notice that the average fermionic non-Gaussianity of tensor product uniform states is
\begin{equation}
    \frac{2}{\pi}\int_0^{\pi}\sin(\beta) \Lambda_{{\rm ferm}}(\ket{\psi}) d\b = 1-\frac{2}{2 n+1}\,.
\end{equation}

\subsection{Clifford non-stabilizerness for uniform tensor product states, proof of Proposition~\ref{prop:uent-stab}}

Again, without loss of generality we parameterize any tensor product uniform state as $(R_z(\a)R_y(\b) \ket{0})^{\otimes n}$. 
Now recall $\Lambda_{\rm stab}(\ket{\psi})=\frac{1}{2^n-1}\sum_{P\in\PC}\Tr[\dya{\psi} P ]^4$, where $\PC$ is the set of all $4^n-1$ Pauli operators. 
Hence, substituting the expression for the uniform tensor product state we get
\begin{equation}
    \Lambda_{\rm stab}(\ket{\psi})= \frac{1}{2^n-1}\sum_{P\in\PC}\Tr[(\dya{0})^{\otimes n} (R_y^\dagger(\b) R_z^\dagger(\a))^{\otimes n} P (R_z(\a) R_y(\b))^{\otimes n}]^4\,.
\end{equation}

We can further simplify the calculations by using the identity $\Tr[A \otimes B] = \Tr[A]\Tr[B]$ so that, calling $P_j$ the local Pauli terms appearing in the Pauli string $P$, the full equation becomes
\small
\begin{equation}
    \Lambda_{\rm stab}(\ket{\psi})=\frac{1}{2^n-1}\sum_{P\in\PC} \left(\prod_{j=1}^n\Tr[\left(\frac{\id + Z}{2}\right)(c_\b \id+i s_\b Z)(c_\a \id+i s_\a Y)P_j(c_\b \id-i s_\b Z)((c_\a \id-i s_\a Y)]\right)^4\,,
\end{equation}
\normalsize
where $c_\th=\cos(\th/2)$ and $s_\th=\sin(\th/2)$.
Carrying out the trace one finds, for a given Pauli $P$ with $m_0$ $\id$ operators, $m_x$ $X$ operators, $m_y$ $Y$ operators and $m_z$ $Z$ operators, a contribution equal to $(\cos (\beta)\sin(\alpha))^{4 m_x} (\sin(\beta)\sin (\alpha))^{4 m_y}\cos(\alpha)^{4m_z}$. 
We can then carry out the summation by adding up this contribution for every combination of $m_0$, $m_x$, $m_y$ and $m_z$ which results in
\begin{equation}
    \Lambda_{\rm stab}(\ket{\psi})=\frac{1}{2^n-1}\sum_{m_0=0}^{n-1}\sum_{m_x=0}^{n-m_0}\sum_{m_y=0}^{n-m_0-m_x}\binom{n}{m_0}\binom{n-m_0}{m_x}\binom{n-m_0-m_x}{m_y} (c_\b s_\a)^{4 m_x} (s_\b s_\a)^{4 m_y} c_\a^{4 (n - m_0-m_x-m_y)}\,.
\end{equation}

This expression simplifies to 
\begin{equation}
    \Lambda_{\rm stab}(\ket{\psi})= \frac{(1 + \cos^4(\b/2) + \frac{1}{4}(3 + \cos(2 \a))\sin^4(\b/2))^n - 1}{2^n-1}.
\end{equation}

Notice that for $\b = 0$ we get $\Lambda_{\rm stab}(\ket{\psi})=1$, as we are essentially calculating the non-stabilizerness of $\ket{0}^{\otimes n}$. Instead, if we set $\a = 0$ we get the Clifford non-stabilizerness formula for a uniform real product state
\begin{equation}
    \frac{(1 + \cos^4(\b/2) + \sin^4(\b/2))^n - 1}{2^n-1}\,.
\end{equation}

\subsection{Gaussian state expectations, proof of Proposition~\ref{prop:haar-gauss}}

Here we provide proofs for the results presented in Proposition~\ref{prop:haar-gauss} for the expected witness values of random Gaussian states. Given that any state in $\SC_{\rm ferm}$ can be expressed as $\ket{\psi}=R(g)\ket{0}^{\otimes n}$, with $R$ the spinor representation and $g\in \text{SO}(2n)$, we will henceforth define the expectation values of random states in $\SC_{\rm ferm}$ as $\mathbb{E}_{\SC_{\rm ferm}}[f(\dya{\psi})]=\mathbb{E}_{g\in \text{SO}(2n)}[f(R(g)\dya{0}^{\otimes n}R(g)\ad)]$. 

Next, we will only evaluate here the witnesses for the QRTs of entanglement, realness and imaginarity. As per the results of  Appendix~\ref{app:weingarten}, in order to use Eq.~\eqref{eq:weingarten} we need a basis for the second order commutant of free-fermionic matchgate unitaries. As shown in~\cite{diaz2023showcasing},  
\begin{align}\label{eq:com-ff}
{\rm comm}^{(t)}(G)={\rm span}_{\mathbb{C}}\{Q_{k}^0,Q_{k}^1\}\,,
\end{align}
where 
\begin{equation}
        Q_{k}^0 = \NC_k \sum_{P'\in L_k} P' \otimes P' \,,\quad 
    Q_{k}^1 = i\NC_k \sum_{P'\in L_k} P' \otimes Z^{\otimes n}P'
\end{equation}
are orthonormal Hermitian operators 
for integers $k\in [2n]$ and $\NC_k = \left[d\sqrt{
\binom{2n}{k})}\right]^{-1}$ and where $L_k$ denotes the set of $\binom{2n}{k}$ Pauli operators that can be expressed as the product of $k$ distinct Majoranas. In particular, we recall that if a given Pauli $P'\in L_k$ with $k$ even (odd), then $P'$ commutes (anticommutes) with the fermionic Parity operator $Z^{\otimes n}$.

Combining Eq.~\eqref{eq:com-ff} with the Weingarten formula of Eq.~\eqref{eq:weingarten}, allows us to find that 
\begin{equation}\label{eq:average}
    \mathbb{E}_{\SC_{\rm ferm}}[\dya{\psi}] = \sum_{i=0}^1\sum_{k=1}^{2n}\Tr[\dya{0}^{\otimes 2n} Q_k^{(i)}]Q_k^{(i)}= \sum_{k=0,2,4}^{2n} \frac{\binom{n}{k/2}}{d^2\binom{2n}{k}}  \left(\sum_{P'\in L_k} P' \otimes P'+i\sum_{P'\in L_k} P' \otimes Z^{\otimes n}P'\right).
\end{equation}
Here we have used the fact that $\dya{0}^{\otimes 2n}$ is of even fermionic parity, and hence only has support on the operators with even $k$, as well as the fact that 
$\ketbra{0}^{\otimes 2n}$ has support on all $\binom{n}{k/2}$ diagonal Pauli operators in $L_k$.

From Eq.~\eqref{eq:average} we obtain that, for QRTs whose resourcefulness witnesses can be expressed as
\begin{equation}
    \Lambda_{\rm qrt}(\ket{\psi})=C\sum_{P\in\PC} \Tr[\dya{\psi}P]^2\,,
\end{equation}
for a set of Pauli operators $\PC$, then
\begin{equation}\label{eq:group-witness-4}
    \mathbb{E}_{\SC_{\rm ferm}}[\Lambda_{\rm qrt}(\ket{\psi})]=C\sum_{k=0,2,4}^{2n} \frac{\binom{n}{k/2}}{\binom{2n}{k}}\dim(\PC\cap L_k)\,,
\end{equation}
where $\PC\cap L_k$ is the subset of Paulis in $\PC$ which can be expressed exactly as the product of $k$ distinct Majoranas.

\subsubsection{Entanglement}

We recall that the set of operators $\PC_{\rm ent}$ contains all single qubit Paulis $\{X_i\}_{j=1}^n$, $\{Y_i\}_{j=1}^n$ and $\{Z_i\}_{j=1}^n$. Since $X_j$ and $Y_j$ anticommute with $Z^{\otimes n}$, they belong to a set $L_k$ with $k$ odd and hence do not contribute to Eq.~\eqref{eq:group-witness-4}. Then, all $Z_j=\gamma_{2j-1}\gamma_{2j}$ for $j=1,\ldots,n$ and hence belong to $L_2$. Using  Eq.~\eqref{eq:group-witness-3} leads to 
\begin{equation}
    \mathbb{E}_{\SC_{\rm ferm}}[\Lambda_{\rm ent}(\ket{\psi})]=\frac{1}{2n-1}\,.
\end{equation}

\subsubsection{Imaginarity}

Here we begin by noting from Eq.~\eqref{eq:maj} that all Majorana operators $\gamma_i$ with $i$ even (odd) are symmetric (antisymmetric). As such, given a set $L_k$ we can find the number of symmetric operators by noting that out of $2n$ Majoranas, $n$ of them are symmetric, and $n$ of them are antisymmetric. Then, when multiplying $k$ Majoranas, one needs to take into account that reversing their product also generates minus signs (as Majoranas anticommute). This leads to  $\frac{1}{2} \left(\binom{n}{\frac{k}{2}}+\binom{2 n}{k}\right)$ symmetric operators in $L_k$. 
We thus find 
\begin{equation}
    \mathbb{E}_{\SC_{\rm ferm}}[\Lambda_{\rm imag}(\ket{\psi})]=\frac{1}{2^n-1}\sum_{k=2,4,\ldots}^{2n}\frac{\left(\binom{n}{\frac{k}{2}}+\binom{2 n}{k}\right) \binom{n}{\frac{k}{2}}}{2 \binom{2 n}{k}}\,.
\end{equation}

\subsubsection{Realness}

As before, we need to count how many antisymmetric Paulis there are in $L_k$. A straightforward calculation reveals that there are $\frac{1}{2} \left(\binom{n}{\frac{k}{2}}-\binom{2 n}{k}\right)$ of them, which leads to 
\begin{equation}
    \mathbb{E}_{\SC_{\rm ferm}}[\Lambda_{\rm real}(\ket{\psi})]=\frac{1}{2^{n-1}}\sum_{k=2,4,\ldots}^{2n}\frac{\left(\binom{2 n}{k}-\binom{n}{\frac{k}{2}}\right) \binom{n}{\frac{k}{2}}}{2 \binom{2 n}{k}}\,.
\end{equation}

\subsubsection{Non-uniform entanglement}

Resorting again to Eq.~\eqref{eq:group-witness-4}, and to the fact that only the $Z_i\otimes Z_i$ terms will contribute, we readily find that find that
\begin{equation}
    \mathbb{E}_{\SC_{\rm ferm}}[\Lambda_{\rm uent}(\ket{\psi})]=\frac{1}{n(2n-1)}\,.\end{equation}

\subsection{Resourcefulness of stabilizer states, proof of Proposition~\ref{prop:cliff-values}}

Here we provide proofs of Proposition~\ref{prop:cliff-values} which state that Clifford states take discrete values for various resource witnesses whose bases consist of Paulis. Note that these proofs simply offer sets of forbidden values and do not necessarily imply that the remaining values are not also forbidden.

In what follows, we will express  a general stabilizer state $\ket{\psi}\in\SC_{\rm stab}$ as $\ket{\psi}=U\ket{0}^{\otimes n}$, with $U\in\mathbb{G}_{\rm stab}$.  Furthermore, we will call $\rho=\dya{\psi}=U\dya{0}^{\otimes n}U^\dagger$ its associated density matrix.

\subsubsection{Entanglement}

Here we provide proofs for the resourcefulness of stabilizer states $\ket{\psi}$ from the point of view of the multipartite entanglement QRT. First of all let us show that the entanglement witness $\Lambda_{{\rm ent}}(\ket{\psi}) = \frac{1}{n}\sum_{P\in\PC_{\rm ent}}\Tr[\dya{\psi}P]^2$, where we recall $\PC_{\rm ent}=\cup_{i=1}^n\{X_i, Y_i, Z_i\}$, can only take values in $\{\frac{j}{n}\}_{j=0}^n$. 

Let us start by noting that the $n$-qubit zero state $\dya{0}^{\otimes n}=\left(\frac{\id + Z}{2}\right)^{\otimes n}$ can be expressed as a uniform superposition of all $2^n$ Pauli strings in the set $\{\id, Z\}^{\otimes n}$. Since Clifford unitaries map Pauli operators to Pauli operators up to phases $\pm 1$, any stabilizer state $\dya{\psi}$ is similarly expressed as a uniform superposition (again, modulo phases) of certain Pauli strings. Given the orthogonality of Pauli operators and their Hilbert-Schmidt norm of $2^n$ which exactly cancels the normalization factor of the state one sees immediately that $\Tr[\dya{\psi}P]^2$ equals one if $P$ appears in the decomposition of $\dya{\psi}$, and zero otherwise. Thus, $\Lambda_{{\rm ent}}(\ket{\psi})$ must necessarily be an integer multiple of $1/n$.
Lastly, since the initial set of Paulis $\{\id, Z\}^{\otimes n}$ consists of mutually commuting operators, and no unitary can map commuting operators to non-commuting ones, the maximum value of $\Lambda_{{\rm ent}}(\ket{\psi})$ is determined by the maximal number of mutually commuting elements in $\PC_{\rm ent}$. Recalling that the operators in $\PC_{\rm ent}$ form a basis for the induced representation of the algebra $\bigoplus_{i=1}^n\mathfrak{su}(2)$, this number is given by the dimension of its Cartan subalgebra. One finds that the Cartan subalgebra has size $n$ (a standard choice being $\{Z_i\}_{i=1}^n$), thus completing the proof.

Next, let us show that $\Lambda_{{\rm ent}}(\ket{\psi})$ can never take the value $\frac{n-1}{n}$. To see this, observe that any non-vanishing contribution to $\Lambda_{{\rm ent}}(\ket{\psi})$ implies that $\dya{\psi}$ contains a local Pauli $P \in \PC_{\rm ent}$. Without loss of generality, assume that $Z_1$ appears in the decomposition of $\dya{\psi}$. Then, since all Pauli operators in the decomposition must commute, $X_1$ and $Y_1$ cannot appear. This, in turn, implies that the reduced density matrix over the first qubit is $\rho_1 = \frac{\id + Z}{2}$, i.e., $\rho_1 = \dya{0}$ is pure. The same reasoning applies if $X_1$ or $Y_1$ appears instead. Hence, for any non-zero contribution to $\Lambda_{{\rm ent}}(\ket{\psi})$ from a local Pauli, the corresponding qubit must be in a separate pure state. If it were possible to find exactly $n - 1$ such contributions, it would follow that $\dya{\psi} = (\bigotimes_{j=1}^{n-1} \dya{\phi_j}) \otimes \sigma$, for some set of pure states $\ket{\phi_j}$ (eigenvectors of $X_j$, $Y_j$, or $Z_j$), and a single-qubit state $\sigma$. The latter would then necessarily be maximally mixed, $\sigma = \frac{\id}{2}$, since by assumption no local Pauli on the last qubit appears in the decomposition of $\dya{\psi}$. However, this leads to a contradiction, as $\dya{\psi}$ is pure by definition, thus proving that $\Lambda_{{\rm ent}}(\ket{\psi})$ cannot take the value $ \frac{n-1}{n}$.

\subsubsection{Fermionic non-Gaussianity}

We now consider the resourcefulness of stabilizer states $\ket{\psi}$ from the point of view of the fermionic non-Gaussianity QRT. Recall that $\Lambda_{{\rm ferm}}(\ket{\psi})=\frac{1}{n}\sum_{P\in\PC_{{\rm ferm}} }\Tr[\dya{\psi}P]^2$, where $\PC_{{\rm ferm}}=\{i\gamma_i\gamma_j\}_{1\leq i<j\leq 2n}$ is a representation of the $\mathfrak{so}(2n)$ algebra. This set is spanned by Pauli operators corresponding to all combinations of products of two distinct Majorana operators $\{\gamma_i\}_{i=1}^{2n}$ as defined in Eq.~\ref{eq:maj}. 

Using the same arguments as we did for the entaglement QRT case, one can again show that $\Lambda_{{\rm ferm}}(\ket{\psi})$ can only take values in the discrete set $\{\frac{j}{n}\}_{j=1}^n$. Indeed, the Cartan subalgebra of $\mathfrak{so}(2n)$ has again size $n$ (and a common choice would again be $\{Z_j=i\gamma_{2j-1}\gamma_{2j}\}_{j=1}^n$). 

Now we prove that no stabilizer state $\ket{\psi}$ can have $\Lambda_{{\rm ferm}}(\ket{\psi})=\frac{n-2}{n}$.
To do so, let us first prove the following lemma
\begin{lemma}\label{lemma:1}
    Let $\{ P_j \}_{j = 1}^m\subset\PC_{{\rm ferm}}$, with $m\leq n$, be a set of mutually commuting Pauli operators. Then, there exists a unitary $W \in \mathbb{G}_{\rm stab} \cap \mathbb{G}_{\rm ferm}$ that maps $\{P_j\}_{j=1}^m$ to the standard set $\{Z_j\}_{j=1}^m$.
\end{lemma} 

\begin{proof}
    By definition, unitaries $W$ in $\mathbb{G}_{\rm stab} \cap \mathbb{G}_{\rm ferm}$ preserve both the Pauli group and the number of Majorana modes. Thus, given any Majorana operator $\gamma_\alpha$, which we recall to be proportional to a Pauli, there exists a unitary $W$ such that $W \gamma_\alpha W^\dagger \propto \gamma_\beta$ for any other $\gamma_\beta$, where the proportionality factor is a phase. Namely, the group $\mathbb{G}_{\rm stab} \cap \mathbb{G}_{\rm ferm}$ acts as a signed permutations on the set of Majorana operators $\{\gamma_\a\}_{\a=1}^{2n}$.

    Now consider a set of commuting Pauli operators $\{P_j = i \gamma_{\alpha_j} \gamma_{\beta_j}\}$ belonging to $\PC_{\rm ferm}$. Since for any pair $(j,k)$ we have by assumption that $[P_j,P_k]=0$, and since Majorana operators anticommute, we must have that the pairs $(\alpha_j, \beta_j)$ and $(\alpha_k, \beta_k)$ are disjoint for all $j \ne k$. Thus, each Majorana operator appears at most one time in $\{P_j\}$.

    Applying the unitary $W$, we find that
    \begin{equation}
    W P_j W^\dagger = W i\gamma_{\alpha_j} \gamma_{\beta_j} W^\dagger = i(W\gamma_{\alpha_j} W^\dagger)(W \gamma_{\beta_j} W^\dagger)\,.
    \end{equation}
    By the previous argument, we can choose $W$ such that $\gamma_{\alpha_j}$ and $\gamma_{\beta_j}$ are mapped to $\gamma_{2j-1}$ and $\gamma_{2j}$, respectively, for each $j = 1, \dots, m$. Equivalently, this means $P_j$ is mapped to $i \gamma_{2j-1} \gamma_{2j}$, which corresponds to the Pauli operator $Z_j$ in the Jordan-Wigner representation. Since these $Z_i$ mutually commute and form the standard Cartan subalgebra, it must be that $m\leq n$, proving the lemma.
\end{proof}

Now, assume that a given stabilizer state $\ket{\psi}$ has exactly $\Lambda_{{\rm ferm}}(\ket{\psi}) = \frac{n-2}{n}$. By analogous reasoning to the entanglement QRT case, this implies that in the decomposition of $\dya{\psi}$ as a uniform superposition of commuting Pauli operators, exactly $n-2$ of them belong to $\PC_{\rm ferm}$. Since we are always free to conjugate the state by a free operation from the group $\mathbb{G}_{\rm ferm}$, we can, by the lemma above, choose $W \in \mathbb{G}_{\rm stab} \cap \mathbb{G}_{\rm ferm}$ such that those $n-2$ commuting Pauli operators are mapped to $\{Z_j\}_{j=1}^{n-2}$. As in the entanglement case, this implies that the state transforms as $W\dya{\psi}W^\dagger = \dya{0}^{\otimes n-2} \otimes \sigma$, where $\sigma$ is a state on the last two qubits. Since $W$ is a Clifford and the Clifford group is closed under composition, the state $\sigma$ must itself be a stabilizer state. Thus, we can write $\sigma = \bar{U} \dya{0}^{\otimes 2} \bar{U}^\dagger$ for some $\bar{U} \in C_2$.

We now show that no two-qubit stabilizer state $\ket{\phi} = \bar{U} \ket{0}^{\otimes 2}$ can have $\Lambda_{\rm ferm}(\ket{\phi}) = 0$. Let $\{\bar{\gamma}_a\}_{a=1}^4$ denote the four Majorana operators on two qubits. Then $\dya{\phi} = \frac{1}{4} \prod_{j=1}^2 (\id + i \bar{U} \bar{\gamma}_{2j-1} \bar{\gamma}_{2j} \bar{U}^\dagger)$. Hence, in order for $\ket{\phi}$ to have zero $\Lambda_{\rm ferm}(\ket{\phi})$, one must have that none of the $\bar{U} \bar{\gamma}_1 \bar{\gamma}_2 \bar{U}^\dagger$, $\bar{U} \bar{\gamma}_3 \bar{\gamma}_4 \bar{U}^\dagger$, or their product $\bar{U} \bar{\gamma}_1 \bar{\gamma}_2 \bar{\gamma}_3 \bar{\gamma}_4 \bar{U}^\dagger$ belong to $\PC_{\rm ferm}$. Recall that these three operators mutually commute.

Let $\LC_k$ denote the span of products of $k$ distinct Majorana operators (which together, form a  complete basis for operator space). Then $\PC_{\rm ferm} \subset \LC_2$. If $\bar{U} \bar{\gamma}_1 \bar{\gamma}_2 \bar{U}^\dagger \notin \LC_2$, assume it belongs to $\LC_4$. But $\LC_4$ contains only one nontrivial (up to phase) element: $\bar{\gamma}_1 \bar{\gamma}_2 \bar{\gamma}_3 \bar{\gamma}_4$. Then $\bar{U} \bar{\gamma}_3 \bar{\gamma}_4 \bar{U}^\dagger$ must lie in $\LC_3$ or $\LC_1$, since it cannot be in $\LC_2$ and we have exhausted $\LC_4$. However, it cannot lie therein either, as no product of an even number of Majoranas can commute with a product of an odd number. Thus, neither element can lie in $\LC_4$.

Now consider the case where both elements lie in $\LC_3$. Then, since they differ by exactly one Majorana, they cannot commute. If one lies in $\LC_3$ and the other in $\LC_1$, then either their product lies in $\LC_2$ (contradicting our assumption), or they fail to commute. Finally, if both lie in $\LC_1$, they again do not commute.

We have thus exhausted all possible cases, and we found a contradiction in each of these cases. Therefore, it is impossible for a two-qubit stabilizer state to have zero support on $\LC_2$, leading to a contradiction. This proves that $\Lambda_{{\rm ferm}}(\ket{\psi}) \neq \frac{n-2}{n}$.

\subsubsection{Imaginarity and realness}

We now turn to the case of the imaginarity and realness QRT. In particular, we focus on the case of realness, as the results for imaginarity can be derived from the latter.

Recall that for the realness QRT, the resource witness of an $n$-qubit quantum state is given by $\Lambda_{{\rm real}}(\ket{\psi}) = \frac{1}{2^{n - 1}} \sum_{P \in \PC_{{\rm asym}}} \Tr[\dya{\psi} P]^2$, where $\PC_{{\rm asym}} = \{P \in \{\id, X, Y, Z\}^{\otimes n} \,|\, P = -P^T\}$ denotes the set of antisymmetric Pauli operators, i.e., those consisting of an odd number of $Y$'s.

First of all, let us prove the following fact about pure quantum states.
\begin{lemma}
    For any pure $n$-qubit quantum state $\ket{\psi}$, it holds that
    \begin{equation}
        \sum_{P \in \PC_{{\rm asym}}} \Tr[\dya{\psi} P]^2 \leq 2^{n-1}\,.
    \end{equation}
\end{lemma}
\begin{proof}
    Let us begin by noting that, for $\PC$ the set of all $4^n$ $n$-qubit Pauli operators (including the identity), we have
    \begin{equation}
    \sum_{P \in \PC} \Tr[\dya{\psi} P]^2 = 2^n \Tr[\dya{\psi}^2] = 2^n\,.
    \end{equation}
    This follows from the standard trace identities $\Tr[A]^2 = \Tr[A^{\otimes 2}]$ and $\Tr[(A \otimes B){\rm SWAP}] = \Tr[AB]$, where ${\rm SWAP}$ denotes the operator that exchanges the Hilbert spaces $\HC_A$ and $\HC_B$ on which $A$ and $B$ act, respectively. Additionally, one has the identity ${\rm SWAP} = \frac{\sum_{P \in \PC} P^{\otimes 2}}{2^n}$ for two copies of an $n$-qubit Hilbert space.
    Now, since $\PC = \PC_{{\rm sym}} \cup \PC_{{\rm asym}}$, we obtain the normalization condition
    \begin{equation}
        \sum_{P \in \PC_{{\rm sym}}} \Tr[\dya{\psi} P]^2 + \sum_{P \in \PC_{{\rm asym}}} \Tr[\dya{\psi} P]^2 = 2^n\,.
    \end{equation}
    Next, consider the quantity
    \begin{equation}
    \sum_{P \in \PC} \Tr[\dya{\psi} P] \Tr[\dya{\psi} P^T] = 2^n \Tr[\dya{\psi} (\dya{\psi})^T] \geq 0\,,
    \end{equation}
    which follows from the same trace properties, together with $\Tr[A^T] = \Tr[A]$, and the fact that $\dya{\psi}$ is positive semidefinite. Decomposing over symmetric and antisymmetric Pauli operators, we find
    \begin{align}
        &\sum_{P \in \PC_{{\rm sym}}} \Tr[\dya{\psi} P] \Tr[\dya{\psi} P^T]\,\, + \sum_{P \in \PC_{{\rm asym}}} \Tr[\dya{\psi} P] \Tr[\dya{\psi} P^T]\nonumber\\
        &
        = \sum_{P \in \PC_{{\rm sym}}} \Tr[\dya{\psi} P]^2 - \sum_{P \in \PC_{{\rm asym}}} \Tr[\dya{\psi} P]^2 \geq 0\,.
    \end{align}
    Adding the inequality above to the normalization condition yields
    \begin{equation}
    \sum_{P \in \PC_{{\rm sym}}} \Tr[\dya{\psi} P]^2 \geq 2^{n-1}\,,
    \end{equation}
    and hence
    \begin{equation}
    \sum_{P \in \PC_{{\rm asym}}} \Tr[\dya{\psi} P]^2 \leq 2^{n-1}\,,
    \end{equation}
    as claimed.
\end{proof}
The previous lemma justifies the choice $2^{n-1}$ as the normalization factor for $\Lambda_{{\rm real}}(\ket{\psi})$, and also explains why in the case of imaginarity $\Lambda_{{\rm imag}}(\ket{\psi}) = \frac{1}{2^{n}-1} \sum_{P \in \PC_{{\rm sym}}} \Tr[\dya{\psi} P]^2$ the minimum possible value is $\frac{2^{n-1}-1}{2^n - 1}$ rather than zero.

Consider again the expansion of an $n$-qubit stabilizer state's density matrix, $\dya{\psi} = \prod_{i = 1}^n \frac{(\id + U Z_i U^\dagger)}{2}$.
Each term $U Z_i U^\dagger$ in the expansion is a Pauli operator, and will be either symmetric or antisymmetric, depending on whether it contains an even or odd number of $Y$ operators, respectively.

Our goal is to show that $\ket{\psi}$ has either maximal or minimal realness. That is, the number of antisymmetric Pauli operators appearing in the decomposition of $\dya{\psi}$ is either minimal (zero) or maximal ($2^{n - 1}$). 

Note that since the $Z_i$ commute, the transformed operators $U Z_i U^\dagger$ must also commute. Now, for any two commuting Pauli operators $A$ and $B$ with definite symmetry, their product $AB$ also has definite symmetry: it is symmetric if both $A$ and $B$ are symmetric or both are antisymmetric, and antisymmetric if exactly one of them is antisymmetric.

Therefore, if none of the operators $U Z_i U^\dagger$ are antisymmetric, then all $2^n$ terms in the expansion of $\dya{\psi}$ are symmetric, yielding $\Lambda_{{\rm real}}(\ket{\psi}) = 0$. On the other hand, if at least one $U Z_i U^\dagger$ is antisymmetric, the symmetry rule above implies that exactly half of the $2^n$ terms in the decomposition will be antisymmetric, resulting in maximal realness. 

Since these are the only two possible cases, we conclude that any stabilizer state $\ket{\psi}$ must have either minimal or maximal realness.

Lastly, from the equation $\sum_{P \in \PC_{{\rm sym}}} \Tr[\dya{\psi} P]^2 + \sum_{P \in \PC_{{\rm asym}}} \Tr[\dya{\psi} P]^2 = 2^n$ that we used in the proof of the Lemma above, one readily gets $\Lambda_{{\rm imag}}(\ket{\psi}) - \frac{\Lambda_{{\rm real}}(\ket{\psi})}{2^{1-n} - 2} = 1$, which can be used to show that when realness is maximized imaginarity is minimized and vice-versa.

\subsection{Uniform entanglement inequality, proof of Theorem~\ref{theo:uent-ent-bound}}

We here prove the existence of a bounding relation between the witnesses of entanglement $\Lambda_{{\rm ent}}$ and uniform entanglement $\Lambda_{{\rm uent}}$.
Recall that for any $n$-qubit pure state $\ket{\psi}$, we defined $\Lambda_{{\rm ent}}(\ket{\psi}) = \frac{1}{n}\sum_{P\in \PC_{\rm ent}} \Tr[\dya{\psi} P]^2$ for $\PC_{\rm ent}$ the set of local (i.e., weight one) Pauli operators, and $\Lambda_{{\rm uent}}(\ket{\psi}) = \frac{1}{n^2}\sum_{P\in\PC_{\rm uent}} \Tr[\dya{\psi} P]^2$ with $\PC_{\rm uent} = \{ \sum_iX_i, \sum_iY_i, \sum_iZ_i\}$. 
We now proceed to prove the following theorem
\begin{theorem}
    Given any $n$-qubit pure state $\ket{\psi}$, its resourcefulness with respect to the QRTs of entanglement and uniform entanglement satisfies
    \begin{equation}
         \Lambda_{{\rm uent}}(\ket{\psi}) \leq \Lambda_{{\rm ent}}(\ket{\psi})\,.
    \end{equation}
\end{theorem}

\begin{proof}
    The proof trivially follows from applying the Cauchy-Schwartz inequality to the expression for $\Lambda_{{\rm uent}}(\ket{\psi})$. Indeed, $\Tr[\dya{\psi}(\sum_iX_i)]^2 \le n\sum_{i} \Tr[\dya{\psi}X_i]^2$ with equality if and only if $\Tr[\dya{\psi}X_i] = c_x$ for some constant $c_x$ independent from the site $i$. The same holds for the cases of $P = \sum_iY_i, \sum_iZ_i$, leading to 
    \begin{equation}
        \Lambda_{{\rm uent}}(\ket{\psi}) \le \frac{1}{n^2}\sum_{P_i = X_i, Y_i, Z_i}n\sum_i^{n} \Tr[\dya{\psi}P_i]^2 = \frac{1}{n}\sum_{P\in\PC_{ent}} \Tr[\dya{\psi}P]^2 = \Lambda_{{\rm ent}}(\ket{\psi})\,.
    \end{equation}
\end{proof}

Notice that the equality holds if and only if each of the expectation values $\{\Tr[\dya{\psi}X_i], \Tr[\dya{\psi}Y_i], \Tr[\dya{\psi}Z_i]\}$ is constant and independent of the qubit $i$ on which the given Pauli acts. This occurs, for instance, if $\ket{\psi}$ is $S_n$-equivariant, if $\ket{\psi}$ is a tensor product uniform state, or if $\ket{\psi}$ has zero overlap with any local Pauli operator.

\section{Generating free operations and states}

In this section we briefly describe how the datasets for the numerics were created. For convenience, we recall that our goal is to uniformly sample unitaries from some group of free operations $\mathbb{G}$, and then apply them to the reference state.

\begin{itemize}
\item \textbf{Entanglement, non-uniform entanglement and spin coherence. }
For the QRTs of entanglement, non-uniform entanglement and spin coherence, we need to uniformly sample unitaries from $SU(2)$. Here, we use standard Euler angle $(\alpha,\beta,\eta)$ parametrization (e.g., $e^{-i\alpha S_z}e^{-i\beta S_y}e^{-i\eta S_z} $) with $\alpha,\eta\in[0,2\pi]$, and $\beta\in[0,\pi]$ and sample according to the Haar measure $d\mu=\sin(\b)d\alpha d\beta d\eta$. 
\item  \textbf{Free-fermionic operations.} To make a free Fermionic operator, we first take a matrix $A$, randomly sampled from $\mathbb{SO}(4)$, then we define  $Q = \log(A)$ where $\log$ is the matrix logarithm. From here, we can construct a $2^n\times 2^n$ unitary of the form 
\begin{equation}
    U(A) = \exp(-\frac{1}{2}\sum_{\mu, \nu = 1}^{2n} Q_{\mu, \nu} \gamma_\mu \gamma_\nu)\,,
\end{equation}
where $\gamma_\mu$ represents a Majorana operator defined in the main text. Importantly, one can also directly find circuits which implement these random unitaries via the results in~\cite{braccia2025optimal}.
\item \textbf{Haar random unitaries from $\text{U}(2^n)$ and $\text{O}(2^n)$ }. Given that we work with small problem sizes, we directly generate $2^n\times 2^n$ matrices with SciPy~\cite{virtanen2020scipy}.
\item \textbf{Random Clifford unitaries.} Here we use the built-in Qiskit~\cite{qiskit} package to sample random unitaries from the $n$-qubit Clifford group.
\item \textbf{$S_n$-equivariant circuits}. We generate unitaries by randomly sampling parameters in the gates of the circuits presented in~\cite{schatzki2022theoretical}. In particular, we pick a depth which scales as $\OC(n^3)$. While it is known that such circuit constructions are not universal~\cite{kazi2023universality} given that they are composed of local gates, this is not an issue in our application as we are skipping global unimportant phases. 

\end{itemize}

\end{document}